\documentclass[12pt,reqno,a4paper]{amsart}
\usepackage[left=1.25in,right=1.25in,top=1.25in,bottom=1.25in]{geometry}
\usepackage{graphicx}
\usepackage{color}
\usepackage{amsmath, amsthm, amssymb, amsfonts, mathtools}
\usepackage{accents}
\usepackage{mathrsfs}
\usepackage{natbib}
\usepackage{setspace}
\newtheorem{theorem}{Theorem}
\newtheorem{lemma}{Lemma}
\newtheorem{proposition}{Proposition}
\theoremstyle{definition}
\newtheorem{definition}{Definition}

\newtheorem{remark}{Remark}
\newcommand{\eps}{\varepsilon}
\newcommand{\ul}{\underaccent{\bar}}
\newcommand{\ol}{\bar}
\newcommand{\df}{\mathrm{d}}

\newcommand{\bdis}{\begin{displaymath}}
\newcommand{\edis}{\end{displaymath}}
\newcommand{\beq}{\begin{equation}}
\newcommand{\eeq}{\end{equation}}
\newcommand{\bea}{\begin{eqnarray*}}
\newcommand{\eea}{\end{eqnarray*}}
\newcommand{\bean}{\begin{eqnarray}}
\newcommand{\eean}{\end{eqnarray}}

\newcommand{\Tau}{\mathcal{T}}

\newcommand{\E}{\mathbb{E}}

\DeclareMathOperator*{\argmax}{arg\,max}
\DeclareMathOperator*{\argmin}{arg\,min}

\usepackage{hyperref}
\input{tcilatex}

\begin{document}

\def\uppercasenonmath#1{} 
\let\MakeUppercase\relax 

\title[\uppercase{Compromise, Don't Optimize}]{\larger Compromise, Don't Optimize:\\
Generalizing Perfect Bayesian Equilibrium\\ to Allow for Ambiguity}
\author[\uppercase{Schlag and Zapechelnyuk}]{%
\larger \textsc{Karl H.~Schlag and Andriy
Zapechelnyuk}}
\date{\today}
\thanks{ \ \\
\textit{Schlag}: Department of Economics, University of Vienna, Oskar-Morgenstern-Platz 1,
1090 Vienna, Austria. \emph{E-mail:} karl.schlag@univie.ac.at. \\
\textit{Zapechelnyuk}: School of Economics and Finance, University of St Andrews, Castlecliffe, the Scores, St Andrews KY16 9AR, UK. {\it E-mail:} {az48@st-andrews.ac.uk.} \\ 
\ \\
We are grateful to Pierpaolo Battigalli, Simon Grant, and Clara Ponsati for their comments.
}

\begin{abstract} 
We introduce a solution concept for extensive-form games of incomplete information in which players need not assign likelihoods to what they do not know about the game. This is embedded in a model in which players can hold multiple priors. Players make choices by looking for compromises that yield a good performance under each of their updated priors. Our solution concept is called perfect compromise equilibrium. It generalizes perfect Bayesian equilibrium. We show how it deals with ambiguity in Cournot and Bertrand markets, public good provision, Spence's job market signaling, bilateral trade with common value, and forecasting.

\bigskip

\noindent\emph{JEL\ Classification:}\ D81, D83\newline


\end{abstract}

\maketitle

\newpage
\section{Introduction}
Modeling lack of information is at the center stage of economics. Agents might not know the previous choice of someone else. Or they might not know the type of their opponent. Or they might not know their own payoffs.

The concept of {\it perfect Bayesian equilibrium} (PBE) deals with uncertainty by forcing players to specify priors that describe precise likelihoods of the possible states of the world. However, being uncertain seems to contradict any ability to assign such probabilities. Players might be ambiguous and not willing or able to specify probabilities. They might only be able to identify which states are possible.

In this paper we introduce a solution concept that does not force players to formulate priors. Our solution concept is called {\it perfect compromise equilibrium} (PCE). 
It applies to extensive-form games of incomplete information. Players are allowed to be ambiguous about what they do not know about the game. Ambiguity is modeled by allowing each player to hold a set priors. A player without probability assessments is one that has a set of degenerate priors. A standard Bayesian player is one that has a unique prior. PCE includes PBE as a special case when each player in the game is endowed with a single prior. 
Our solution concept is readily defined once we have resolved the following two issues.  How to learn from the past? How to model decision making under ambiguity?

Learning from the past is modeled as follows. Each player starts a game with a set of priors. These priors are then updated, prior by prior, using Bayes' rule whenever possible \cite[known as full Bayesian updating,][]{Pires}. This determines the player's posterior beliefs at each of her information sets. As in PBE, there are no restrictions on how a prior is updated at information sets that it does not reach.

Decision making under ambiguity is modeled as follows. A player makes a decision at each of her information sets by choosing a {\it best compromise} given the set of her posterior beliefs at that information set.
This is a decision that balances the loss of not making the optimal
decision for each of her beliefs. This criterion collapses to 
expected utility maximization if there is only one belief or if there is
a dominant action at the given information set. The concept of best
compromise follows the tradition of minimax regret and is founded on many pillars. It has an axiomatic foundation.
It is similar to classic expected utility maximization when there is little 
ambiguity, in the sense that all beliefs are close to each other, or when the loss of not making the optimal decision is small for each of the beliefs. Best compromises can be used to justify behavior in front of
people with different preferences. They formalize the everyday notion of making a compromise.

We assume that a player's choice at each of her information sets is made given the equilibrium behavior of herself and others at all other information sets. In particular, the player anticipates her own choices at subsequent information sets, and hence follows {\it consistent planning} \citep{Strotz56,Siniscalchi2011}.

Formally, our solution concept, PCE, specifies for each player a strategy and a belief mapping. The strategy identifies the action the player chooses at each of her information sets. The belief mapping maps each prior of the player to a belief over decision nodes in each of her information sets. We show that a PCE exists in finite games. We illustrate the PCE concept in a simple game that involves a market of lemons with quality inspections. This illustration demonstrates how multiple priors are updated and highlights differences in the reasoning as compared to PBE.

We are particularly interested in modeling players who  have difficulty forming priors, or who are extremely ambiguous and only focus on which states are possible, without assessing their likelihoods. For instance, it seems unlikely that firms conjecture a specific probability distribution when they think about what demand they will be facing. Yet it seems plausible that they put bounds on the uncertain demand. These bounds can come from the most optimistic and pessimistic scenarios provided by expertise.
Situations like this can be modeled within our framework by letting the set of priors consist only of degenerate priors. We call this {\it genuine ambiguity}. This way of modeling incomplete information without using priors comes with
numerous advantages in comparison to PBE. Solutions are often easier
to obtain. They are more parsimonious as they do not change with a
prior. Solutions can be more intuitive as they are simple and depend on observables and
not on fictitious distributions. These advantages are demonstrated in our examples.

We investigate six salient economic examples. We consider Cournot competition with unknown demand, where firms postulate bounds on the true demand. We consider Bertrand competition where firms assess lower and upper bounds on the marginal costs of their rivals. We consider public good provision 
where beneficiaries of a public good do not know each others' values and hypothesize an interval where these values can be. We consider Spence's job market where employers are uncertain about the cost of education and the productivity of workers, and conjecture bounds on these parameters. We consider bilateral trade under common value where each party knows an interval that contains the true value. Finally, we consider forecasting of a random variable with unknown distribution.

These examples highlight the value of the PCE concept in terms of realism, tractability, and new insights. They are arguably more realistic than those found in the literature,
as we do not have to confine ourselves to parametric models of uncertainty or to models with two states (high and low). These examples involve strategic decision making under rich uncertainty where the PBE analysis is intractable. New insights appear. We find that replacing priors by bounds on uncertain parameters has little impact on profits in Cournot and Bertrand competition settings where compromise values are small. In these contexts it makes little sense to think in more detail about which state is really the true one, as payoffs would only be slightly higher in some states but could be substantially lower 
in other states. Yet loosening these bounds causes firms to react differently. They become more competitive under Cournot competition and less competitive under Bertrand competition. In the public good game, we show the ease of comparing policies and the simplicity of the beneficiaries' contribution rules. In the separating equilibrium of Spence's job market signaling game, better educated workers are not necessarily more productive, unlike in the classic model with two types \citep{Spence73}. In bilateral trade with common value, we find that trade is possible, as opposed to
the famous no-trade theorem for PBE \citep{MilgromStokey82}. The possibility that the trading partners
have different valuations leads to trade with positive probability in a PCE, as
ignoring this possibility generates losses that the traders want to minimize. 
Finally, when forecasting a random variable with a known mean and unknown distribution based on a noisy signal, the best-compromise forecast is a weighted average of the mean and the signal.

\vspace{6pt}
\noindent \textbf{Related Literature.} Our paper contributes to the literature
on robustness and ambiguity in games of incomplete information. 

A paper that at a glance may seem very similar to ours is \cite{Hanany2018}. They also consider general extensive form games with incomplete information. Their players have smooth ambiguous
preferences \cite[see also][]{Klibanoff2005}. Specifically, a player combines or aggregates different possible priors
into a single belief using a distribution over these priors and a concave
aggregator function. This aggregated belief is updated over time in a dynamically
consistent fashion. Thus, a player has a very detailed understanding of how the different
priors should be weighted. In contrast, the different priors in our
model remain conceptually separated. The inability or unwillingness to
combine priors is at the heart of our approach. Compromises are chosen
as a way to resolve the conflict of having different possible
understandings of the environment. In fact, one of the emphases of our
paper is that it offers a means to get away from probability assessments. Our approach can handle a player who
wishes to capture uncertainty by assessing a set of possible states,
without any use of priors.

An important ingredient of our solution concept is the use of compromise for making choices when the true state is unknown. 
A popular alternative approach in the literature on ambiguity is maximin preferences \citep{Wald50,Gilboa89}. These preferences have been brought to simultaneous-move games with incomplete information and multiple priors by \cite{Epstein96}, \cite{Kajii97}, \cite{Kajii2005}, and \cite{Azrieli2011}. While the maximin approach can be suitable in applications where players are pessimistic and care about the worst possible payoffs, it leads to unintuitive results in our examples. For instance, in Bertand duopoly with ambiguity about the rival's cost, maximin utility leads firms to shut down. To obtain nontrivial results, additional structural assumptions need to be added, such as assuming knowledge of the mean state. Another approach found in literature is Knightian uncertainty with incomplete preferences. This has been used by Chiesa et al.~(\citeyear{Chiesa}) to model bidding in auctions.

Our idea of best compromise has origins in minimax regret \citep{Savage51} and connects to approximate optimality. Our optimization criterion differs from minimax regret as evaluation occurs at each information set, while minimax regret traditionally evaluates regret ex-post. 
Furthermore, PCE retains the strategic reasoning of PBE, as players have certainty about each others' strategies. For an investigation of minimax regret under strategic uncertainty see \cite{Linhart1989}, and under partial strategic uncertainty see \cite{RenouSchlag2010}.

In simultaneous-move games, PCE can be considered as a generalization of ex-post Nash equilibrium \citep{Cremer1985}. It
can be thought of as an $\varepsilon $-ex-post Nash equilibrium in which the
smallest possible value of $\varepsilon $ is chosen for each player. In the context of $%
\varepsilon $-Nash equilibrium \citep{Radner80} the value of $\eps$ is interpreted a minimal level of improvement necessary to trigger a deviation. Our interpretation is different. The value of $\varepsilon $ measures the compromise needed to accommodate all beliefs. In
particular, the threshold $\varepsilon $ is endogenous in a PCE.

PCE can be interpreted as a robust version of PBE where
robustness in the sense of \cite{Huber1965} means to make choices that also
perform well if the model is slightly misspecified. Being a compromise, our
suggested strategies perform well under each prior given how others 
make their choices, never doing too badly relative to what could be
achieved under that prior. \cite{Stauber} analyzes the local robustness of PBE to small degrees of
ambiguity about player's beliefs. In particular, players do not adjust their play to this ambiguity, unlike our paper.

We proceed as follows. In Section \ref{s:pce} we introduce our solution concept, prove existence, and demonstrate it in a simple game. In Section \ref{s:examples} we illustrate PCE in six self-contained economic examples. Section \ref{s:concl} concludes. All proofs are in Appendix A. An alternative model of the forecasting example is in Appendix B.

\section{Perfect Compromise Equilibrium}\label{s:pce}

We introduce a solution concept called {\it perfect compromise equilibrium (PCE)}. The concept is formally defined in Section \ref{s:setup}. It is further discussed in Section \ref{s:disc-1} and illustrated by a simple example in Section \ref{s:ex-lemons}. A reader who wishes to be spared with the formalities and seeks to understand the essence of PCE and its applicability can jump to Section \ref{s:examples} that presents self-contained economic examples.

\subsection{Formal Setting}\label{s:setup}
\ Consider a finite extensive-form game described by $(N,\mathcal G,$ $\Omega,(\Pi_1,...,\Pi_n),(u_1,...,u_n))$, where $N=\{1,...,n\}$ is a set of players, $\mathcal G$ is a finite game tree, $\Omega$ is a finite set of states, $\Pi_i\subset\Delta(\Omega)$ is a finite set of priors of player $i$, and $u_i$ is a payoff function of player $i$. Note that we allow players to have different sets of priors. 

The game tree $\mathcal G$ describes the order of players' moves, their information sets, and actions that are available at each information set. It is defined by a set of linked decision nodes and terminal nodes that form a tree. Each decision node is assigned three elements: a player $i$, an information set $\phi_i$, and a set of actions available to player $i$ at that information set. Information set $\phi_i$ is a set of all the decision nodes that player $i$ cannot distinguish. Information sets and action sets satisfy the standard assumptions of games with perfect recall. Let $\phi_0$ be the initial decision node of the game, let $\Phi_i$ be the set of all information sets of player $i$ for each $i\in N$, and let $\Tau$ be the set of terminal nodes of the game.  
Let $A_{\phi_i}$ be a finite set of actions available at an information set $\phi_i$, and let $\mathscr A_{\phi_i}=\Delta(A_{\phi_i})$ be the corresponding set of mixed actions.

In the spirit of \cite{Harsanyi67}, all incomplete information is captured by a move of nature at the beginning of the game. Nature moves only once, at the initial decision node $\phi_0$. An action of nature $\omega$ is called {\it state} and is chosen from the set of states $\Omega$.

The game terminates after finitely many moves at some terminal node, and players obtain payoffs. A payoff function of each player $i\in N$ specifies the payoff $u_i(\tau)$ of player $i$ at each terminal node $\tau\in \Tau$. 

A strategy of player $i\in N$ prescribes a mixed action $s_{\phi_i}\in \mathscr A_{\phi_i}$ for each information set $\phi_i\in\Phi_i$. 
A strategy profile $s$ describes the behavior of all players throughout the game. 

Like in Bayesian games, we also specify posterior beliefs of the players in their information sets. Unlike in Bayesian games, each player may have multiple beliefs in each of her information sets. These beliefs are derived from the set of priors, prior by prior, following Bayes' rule whenever possible. This procedure is known as full Bayesian updating \citep{Pires}. We refer to Section \ref{s:ex-lemons} for an example that illustrates this updating.

Formally, for each player $i$ and each information set $\phi_i\in\Phi_i$, let $\beta_{\phi_i}:\Pi_i\to\Delta(\phi_i)$ be a belief mapping that associates each prior $\pi_i\in \Pi_i$ of player $i$ with a posterior probability distribution $\beta_{\phi_i}$ over the decision nodes in $\phi_i$. Thus, in the information set $\phi_i$, player $i$ faces a set $B_{\phi_i}$ of posterior beliefs derived from the set of priors $\Pi_i$, where
\[
B_{\phi_i}=\left\{\beta_{\phi_i}(\pi_i):\pi_i\in\Pi_i\right\}.
\]
We will refer to $B_{\phi_i}$ as the set of {\it beliefs} at $\phi_i$, and to the profile $\beta=(\beta_{\phi_i})_{\phi_i\in\Phi_i,i\in N}$ as the {\it belief system}.

Like in PBE, we will require consistency of beliefs. 

\begin{definition}\label{def:consistency} 
A belief mapping $\beta_{\phi_i}$ is called {\it consistent under a strategy profile $s$} if for each prior $\pi_i\in\Pi_i$ such that the information set $\phi_i$ is reached with a strictly positive probability under strategy profile $s$, the belief $\beta_{\phi_i}(\pi_i)$ is derived by Bayes rule from $\pi_i$.

A belief system $\beta$ is {\it consistent under a strategy profile $s$} if for each $i\in N$ and each $\phi_i\in\Phi_i$ the belief mapping $\beta_{\phi_i}$ is consistent under $s$.
\end{definition}

Note that our definition of consistency does not impose any discipline on the out-of-equilibrium beliefs. If an information set $\phi_i$ cannot be reached under a given prior $\pi_i$ and a given strategy profile $s$, then every belief $\beta_{\phi_i}(\pi_i)\in\Delta(\phi_i)$ is consistent under $s$. Of course, not every choice of out-of-equilibrium beliefs can be sensible in applications. This is the very same problem that emerged in the context of PBE and gave rise to a vast literature on PBE refinements. This problem is of equally high importance for PCE. However, addressing this problem would take us away from the main messages of this paper. Neither the idea of PCE, nor its properties in the examples considered in this paper change if additional assumptions about out-of-equilibrium beliefs are made. So, we leave this question for future research.

Next we define how decisions are made at an information set $\phi_i$. We fix a strategy profile $s$ and determine how to make a choice at $\phi_i$, while keeping choices at all other information sets fixed. The difficulty of making a decision at $\phi_i$ is that the player does not know which belief in the set of beliefs $B_{\phi_i}$ should be used to evaluate the expected payoff. We resolve this issue by assuming the player chooses a best compromise. This is an action that is never too far from the best action under each belief in $B_{\phi_i}$.

Formally, consider a pair $(s,\beta)$. Denote by $\bar u_i(s_{\phi_i}|\phi_i,s,b_i)$ the expected payoff of player $i$ from choosing a mixed action $s_{\phi_i}\in \mathscr A_{\phi_i}$ in an information set $\phi_i$ under the belief $b_i$ over the decision nodes in $\phi_i$, assuming that the play is given by $s$ elsewhere in the game.
The payoff difference
\[
\sup_{x_i\in \mathscr A_{\phi_i}} \bar u_i(x_i|\phi_i,s,b_i)-\bar u_i(s_{\phi_i}|\phi_i,s,b_i)
\]
is called player $i$'s {\it loss} from choosing mixed action $s_{\phi_i}$ at information set $\phi_i$ given belief $b_i$. It describes how much better off player $i$ could have been at this information set given this belief if, instead of choosing $s_{\phi_i}$, she had chosen the best action, assuming that the actions in all other information sets are prescribed by $s$. The {\it maximum loss} of player $i$ from choosing a mixed action $s_{\phi_i}$ in an information set $\phi_i$ under $(s,\beta)$ is given by
\[
l(s_{\phi_i}|\phi_i,s,\beta)=\max_{b_i \in B_{\phi_i}}\left( \sup_{x_i\in \mathscr A_{\phi_i}} \bar u_i(x_i|\phi_i,s,b_i)-\bar u_i(s_{\phi_i}|\phi_i,s,b_i)\right).
\]
So the maximum is evaluated over all beliefs of player $i$ at $\phi_i$. 

Player $i$ makes a decision that minimizes the maximum loss. Such a choice is called a {\it best compromise}. Formally she chooses an element of
\beq\label{e-pce-opt}
\argmin_{s_{\phi_i}\in\mathscr A_{\phi_i}} l(s_{\phi_i}|\phi_i,s,\beta)
\eeq
at each of her information sets $\phi_i$. In equilibrium $s^*$, this means that she chooses $s^*_{\phi_i}\in\argmin_{s_{\phi_i}\in\mathscr A_{\phi_i}} l(s_{\phi_i}|\phi_i,s^*,\beta)$. Hence, when computing the maximum loss and finding the best compromise, each player assumes that the behavior is given by $s^*$ at all other information sets, including her own. Thus the players anticipate their own choices at subsequent information sets, which is known as {\it consistent planning} \citep{Strotz56,Siniscalchi2011}.

This leads to our equilibrium concept that is based on the ideas of best compromises and consistent beliefs.

\begin{definition}\label{def:pce}
A pair $(s^*,\beta^*)$ is called a {\it perfect compromise equilibrium} if 

\noindent (a) each player chooses a best compromise in each of her information sets;

\noindent (b) the belief system $\beta^*$ is consistent under the strategy profile $s^*$.
\end{definition}

We begin by establishing the existence of PCE. 
\begin{theorem}\label{p:exist}
A perfect compromise equilibrium exists.
\end{theorem}
The proof is in Appendix \ref{s:pt}.

\begin{remark}
In some applications there can be a continuum of strategies, states, and priors over states. The definition of PCE readily extends to such settings, but some additional assumptions have to be made to ensure its existence. 
\end{remark}

\begin{remark}\label{rem-1}
In some applications it can be unrealistic to assume that players choose mixed actions. Our definition of PCE can be easily adjusted if players are only allowed to choose pure actions. In this case, a best compromise means to minimize one's maximal loss among the available pure actions. Formally, we set  $\mathscr A_{\phi_i}=A_{\phi_i}$ for each player $i$ and each information set $\phi_i$, and define the concept of PCE as above.
\end{remark}

In the remainder of this section, we discuss some properties of the PCE and provide a simple example. This example illustrates the PCE concept and outlines the difference from a PBE where each player has a single prior.

\subsection{Discussion}\label{s:disc-1}
We highlight some properties of PCE.

\smallskip\noindent{\bf Best Compromise.} 
Our decision making criterion for how to make choices at a given information set
captures the intuitive notion of making a compromise. As a compromise, the performance should be satisfactory in all potential situations, as opposed to being best under some and possibly very bad under others. The concept of best compromise identifies the smallest maximal distance from first best as a measure of how large the compromise has to be. Compromises are valuable when decisions have to be justified in front of others who have heterogeneous perceptions about the environment. 

The concept of a best compromise follows the tradition of decision making under minimax regret, thus having an axiomatic underpinning \citep{Milnor,Stoye2011}. Traditionally, minimax regret is evaluated ex-post after all uncertainty is resolved. In contrast, to model a compromise in the face of several beliefs, we consider the loss attained at the interim (at a given information set) for a given belief. Stoye's (2011) axioms continue to hold from this interim viewpoint. Furthermore, our concept retains the strategic reasoning of PBE, as players know each others' strategies. This is unlike \cite{Linhart1989} who reduce the game to an individual decision problem, where the behavior of the others is treated as a move of nature.

Clearly, instead of best compromise, any other decision making criterion under ambiguity could be used for determining choices at information sets. For instance, the maximin utility criterion can be used to model pessimism or cautiousness, a world in which the player always anticipates the worst outcome.

\smallskip\noindent{\bf PCE vs PBE.}
Our definition of PCE generalizes the concept of PBE to games where some players may be ambiguous about what they do not know. When there is no ambiguity, so there is a single belief at each information set, then our setting describes a standard game of incomplete information. In this case, the loss minimization objective, as described in \eqref{def:pce}, reduces to the standard utility maximization objective. So, an action minimizes the maximum loss of a player if and only if it is a best response. Moreover, whenever there is only a single belief, the consistency requirement introduced in Definition \ref{def:consistency} reduces to the standard Bayesian consistency of beliefs. Hence, PCE becomes PBE \citep[in the sense of ][]{FT}.

The difference between PCE and PBE emerges in models where some players are ambiguous about the state of the world. The standard PBE approach forces players to quantify the uncertainty by specifying a unique belief at each information set, and then assuming that the players optimize with respect to these beliefs. Our approach sidesteps this issue by letting the players have multiple beliefs at each information set and find compromises with respect to these beliefs.

\smallskip\noindent{\bf Ex-post Nash equilibrium.} 
In simultaneous move games PCE is related to ex-post Nash equilibrium. Ex-post Nash equilibria are profiles that are Nash equilibria in the game in which the state is observed by all players at the outset of the game. This means that the maximum loss of each player at her single information set is equal to zero. Consequently, any ex-post Nash equilibrium is also a PCE. Note, however, that ex-post Nash equilibria often do not exist.

\smallskip\noindent{\bf Dominance.} A PCE survives the elimination of strictly dominated strategies, as we now demonstrate. We say that an action $a_i\in \mathscr A_{\phi_i}$ at an information set $\phi_i$ is {\it strictly dominated} for player $i$ if there exists another action $x_i\in\mathscr A_{\phi_i}$ such that player $i$'s payoff from choosing $a_i$ is strictly worse than that from choosing $x_i$, regardless of the state $\omega\in\Omega$ and of the choices of other players at any of their information sets. Iterated dominance is defined as usual. After having excluded actions that were strictly dominated in previous rounds, one checks the dominance condition w.r.t.~the remaining actions of each player. Now observe that if an action $a_i$ at some information set $\phi_i$ is strictly dominated, then it cannot be a best compromise at this information set. This is because the (mixed) action that strictly dominates $a_i$ will achieve a strictly lower loss for each belief, and hence its maximal loss will be strictly smaller. Thus, a strictly dominated action cannot be a part of a PCE. This argument can be iterated, so any iterated strictly dominated action cannot be a part of a PCE. 

\begin{figure}[!t]
\includegraphics[width=250pt]{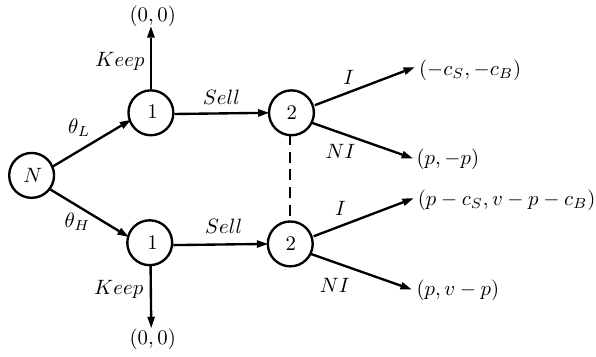}
\caption{Lemon Market with Quality Inspections}\label{F:A}
\end{figure}

\subsection{Example}\label{s:ex-lemons}

The following example illustrates the PCE concept and highlights the difference from PBE.

A seller has a car whose quality is either low ($\theta_L$) or high ($\theta_H$). She observes the quality of the car and decides whether to offer it for sale at a fixed price $p$ or to keep it. If the car is offered for sale, a potential buyer decides whether to buy it without observing its quality, or to demand a costly inspection that reveals the quality. The cost of the inspection is $c_S$ for the seller and $c_B$ for the buyer. The low-type car is worth zero and the high-type car is worth $v$ to the buyer. The seller's value of either type of the car is zero. The payoff parameters $v, p, c_S$, and $c_B$ are commonly known and  assumed to satisfy
\beq\label{e-assum-lemon}
v-c_B\ge p>c_S\ge 2c_B>0.
\eeq
Consequently, the buyer either decides to purchase the car without demanding the inspection, or to ask for the inspection, in which case he buys if it is high quality and does not buy otherwise. So the only nontrivial choice of the buyer is whether or not to inspect the car. The game is summarized in Figure \ref{F:A}, where $N$, $S$, and $B$ denote Nature, the seller, and the buyer, respectively, and $I$ and $NI$ denote the buyer's choice to inspect or not to inspect the car, respectively. 

Suppose that the buyer is uncertain about the quality of the car and holds multiple priors about whether the quality is low or high. For convenience, beliefs are summarized by the probability that the car quality is high. Multiple priors can arise in many different ways. The buyer might come up with different competing scenarios for explaining what car is being offered, each scenario possibly leading to a different prior. In this case the buyer's set of priors is $\Pi_B=\{\pi_0,\pi_1,..,\pi_K\}$. The buyer might have some vague understanding of the likelihoods, for instance that a high quality car is more likely than a low quality car. In this case $\Pi_B=\{\pi: \pi\ge 1/2\}$. The buyer might wish to base her purchasing behavior on the possibility that the quality might be high and that it might be low and does not wish to base her behavior on any likelihoods. This motivates setting $\Pi_B=\{0,1\}$. In what follows we analyze the case $\Pi_B=\{\pi_0,\pi_1,..,\pi_K\}$ where $0=\pi_0<\pi_1<...<\pi_{K-1}<\pi_K=1$. In particular, we are including the two degenerate priors where the quality is low and high with certainty.

The seller has two information sets with single decision nodes. So the seller's beliefs are trivial in these information sets. The buyer has a single information set denoted by $\phi_B$ that contains two decision nodes. The buyer updates each of his priors to obtain a set of beliefs $\{\beta_{\phi_B}(\pi_k)\}_{k=0,1,...,K}$ at this information set, where $\beta_{\phi_B}(\pi_k)$ is the buyer's belief mapping.

We now characterize the PCE of this game. A PCE is summarized by a triple $(\sigma^*_S,\sigma^*_B, \beta^*_{\phi_B})$, where $\sigma^*_S(\theta)$ is the probability that the seller sells the car conditional on its type $\theta\in\{\theta_L,\theta_H\}$, $\sigma^*_B$ is the probability that the buyer inspects the car, and $\beta^*_{\phi_B}$ is the buyer's belief mapping.

Because the seller has no ambiguity, the best compromise for the seller is her best response. When the quality is high, the seller prefers to sell the car, as this is a strictly dominant strategy, so $\sigma_S^*(\theta_H)=1$. When the quality is low, the seller prefers to sell the car if the probability of inspection $\sigma^*_B$ is low enough, and to keep the car otherwise, specifically,
\beq\label{e-BR-S-lemon}
 \sigma_S^*(\theta_L)\in\begin{cases}
\{1\} & \text{if $\sigma^*_B<p/(p+c_S),$}\\ 
[0,1] & \text{if $\sigma^*_B=p/(p+c_S),$}\\ 
\{0\} & \text{if $\sigma^*_B>p/(p+c_S).$} 
 \end{cases}
\eeq
The buyer can be ambiguous, because he can have multiple beliefs in his information set. In order to be a part of a PCE, the buyer's belief mapping must be {\it consistent} with the strategy $\sigma_S^*$ of the seller. Specifically, each prior $\pi_k\in\Pi_B$ is transformed into a belief $\beta_{\phi_B}^*(\pi_k)$ using Bayes' rule whenever possible. Given $\sigma_S^*(\theta_L)\in[0,1]$ and $\sigma^*_S(\theta_H)=1$, the consistent belief mapping is 
\beq\label{e:Bayes}
\beta_{\phi_B}^*(\pi_k)=\frac{\pi_k}{\pi_k+(1-\pi_k)\sigma^*_S(\theta_L)}
\eeq
for each $\pi_k\in \Pi_B$, except for the case of $\pi_k=0$ and $\sigma_S(\theta_L)=0$ where the above Bayes' posterior is undefined. In this case, any belief $\beta_{\phi_B}^*(0)\in[0,1]$ is consistent with the seller's strategy. The set of buyer's beliefs in her information set $\phi_B$ is given by
\[
B_{\phi_B}=\{\beta^*_{\phi_B}(\pi_k)\}_{k=0,1,...,K}.
\]

For each belief $b$, denoting the probability that the car has high quality, the buyer's optimal choice is to inspect when $b<1-c_B/p$, and not to inspect when $b> 1-c_B/p$ (being indifferent when $b=1-c_B/p$). The buyer's loss for a given belief $b$ from a strategy $\sigma_B$ describes how much more payoff the buyer could have obtained if he optimized his choice under this belief. For $b\le 1-c_B/p$ this loss is given by
\begin{align*}
\big[b(v-p)-c_B\big]&-\big[(b(v-p)-c_B)\sigma_B+(bv-p)(1-\sigma_B)\big]\\
&=((1-b)p-c_B)(1-\sigma_B).
\end{align*}
For $b\ge 1-c_B/p$ this loss is given by
\begin{align*}
\big[bv-p\big]&-\big[(b(v-p)-c_B)\sigma_B+(bv-p)(1-\sigma_B)\big]\\
&=(c_B-(1-b)p)\sigma_B.
\end{align*}
The buyer's {\it maximum loss} among his different beliefs of choosing the inspection strategy $\sigma_B$ is thus
\begin{multline}
l_S(\sigma_B|\phi_B,\sigma_S^*,\beta^*_{\phi_B})=\\\max_{b\in B_{\phi_B}} \max\Big\{((1-b)p-c_B)(1-\sigma_B),(c_B-(1-b)p)\sigma_B\Big\}.\label{e:loss-lemon}
\end{multline}
Intuitively, the buyer who has multiple beliefs and anticipates the seller to follow her equilibrium strategy worries about two possible situations. It could be that the probability of high quality is high, so the buyer loses payoff by inspecting. The greatest such loss occurs when the belief is the highest, $b=\beta^*_{\phi_B}(1)$. Alternatively, it could be that the probability of high quality is low, so the buyer is losing payoff by not inspecting. The greatest such loss occurs when the belief is the lowest, so $b=\beta^*_{\phi_B}(0)$. The buyer thus chooses the best-compromise inspection strategy $\sigma_B^*$ that balances these two losses.

\begin{proposition}\label{p:lemon}
A profile $(\sigma^*_S,\sigma^*_B,\beta^*_{\phi_B})$ is a perfect compromise equilibrium if and only if the seller and buyer's strategies $\sigma^*_S$ and $\sigma^*_B$ are given by
\begin{align*}
\sigma^*_S(\theta_L)&=0, \ \ \sigma^*_S(\theta_H)=1, \ \ \sigma^*_B=1-\frac{c_B}{(1-b_0)p},
\end{align*}
and the buyer's belief mapping $\beta^*_{\phi_B}$ is given by
\[
\beta^*_{\phi_B}(0)=b_0, \ \ \beta^*_{\phi_B}(\pi_k)=1 \ \text{for each $k=1,...,K$},
\]
where $b_0$ satisfies
\[
b_0\in\left[0,1-\frac{c_B}{c_S}-\frac{c_B}{p}\right].
\]
\end{proposition}

The proof is in Appendix \ref{s:lemon-proof}.

Our prediction based on the PCE is as follows. High type cars are sold
and low type cars are not. While a rational buyer would infer in this
situation that the quality of the car is high, under a PCE the buyer inspects the car with a positive probability. This happens because the buyer remains doubtful about the quality and needs to find a compromise given her pair of beliefs $\{b_0,1\}$. The specific probability of inspection is the one that yields the best compromise. Note that the buyer also inspects a car with a positive probability under a PBE with a nondegenerate prior. Under a PBE the inspection probability is the one that makes the low-type seller indifferent between selling or not selling.

To see why low type cars are not sold in a PCE, observe that if they were sold
with some positive probability, then the buyer's set of beliefs would include the two degenerate beliefs, $\beta^*_{\phi_B}(0)=0$ and $\beta^*_{\phi_B}(1)=1$. The buyer would then choose the best compromise strategy when facing these two extreme cases. This would then lead to an inspection probability so large that the low type seller would prefer not sell the car. But this would be incompatible with low type cars being sold with positive
probability.

To see why the buyer must have more than a single belief, thus remaining doubtful, consider the following arguments. If there is a single belief, it must be that the car is almost certainly of high quality, as only high quality cars are sold. In this case, the buyer's best compromise is not to inspect the car. But then, the seller would have strictly preferred to sell the car of low quality. 

Note that the belief $b_0=\beta^*_{\phi_B}(0)$ obtained under the degenerate prior $\pi_0=0$, where the car almost certainly has low quality, is not pinned down by Bayes's rule. This is because it is formed in an out-of-equilibrium event, in which the low-type car is offered for sale. As long as $b_0$ is not too high, the buyer's best-compromise inspection probability is high enough to keep the low type seller out of the market.

\section{Examples}\label{s:examples}

We illustrate our solution concept in a few economic examples that are prominent in the literature. 
We consider Cournot and Bertrand duopoly, public good provision, Spence's job market signaling, bilateral trade with common value, and forecasting. The examples presented in this section are self-contained as they do not require knowledge of the formalities presented in Section \ref{s:pce}.

We are particularly interested in understanding strategic play under uncertainty when the players cannot or are unwilling to assess the likelihood of different states of the world at the beginning of the game. Formally, players can only have degenerate priors that put probability one on a single state of the world. We call this {\it genuine ambiguity}. 

Apart from forecasting, the examples presented below deal with genuine ambiguity.
Therein, ambiguity is specified in terms of bounds on what the players do not know. Probability distributions do not play a role. Players do not have beliefs. Instead, they speculate about which state is true or about what decision node within an information set they are at. In addition, we assume that players do not use mixed strategies. They search among their pure strategies for a best compromise. Thus we perform a strategic analysis without using probabilities.

The section concludes with an example of forecasting that demonstrates the interplay between ambiguity and noise, where multiple priors over one parameter meet a single prior over another parameter.

\subsection{Cournot Duopoly with Unknown Demand}\label{s:cournot}

We investigate how two firms compete in quantities when neither firm knows the demand. 

There are two firms that produce a homogeneous good. For clarity of exposition, we assume that there are no costs of production.
Each firm $i=1,2$ chooses a number of units $q_i\ge 0$ to produce. Choices are made simultaneously. The firms face an inverse demand function $P(q_1+q_2)$. Firm $i$'s profit is given by
\[
u_i(q_i,q_{-i}; P)= P(q_i+q_{-i})q_i,  \ \ i=1,2.
\]
Neither firm knows the inverse demand $P$, but they know that it belongs to a set $\mathcal P$ given as follows. Let 
\[
\begin{split}
&\ul P(q)=\ul a-\ul b q \ \   \text{and} \ \ \ol P(q)=\ol a-\ol b q, \ \ \text{where} \ \ \ol a\ge\ul a>0 \ \  \text{and} \ \ \ol a/\ol b\ge \ul a/\ul b>0.
\end{split}
\]
Let $\mathcal P$ be the set of inverse demand functions that satisfy
\beq\label{e:demand}
\begin{split}
&\text{$P(q)$ is continuously differentiable in $q$,}\\
& \ul P(q) \le  P(q)\le  \ol P(q) \quad\text{and}\quad \ul  P'(q)\le  P'(q)\le \ol  P'(q).
\end{split}
\eeq

A firm $i$'s {\it maximum loss} of choosing quantity $q_i$ when the other firm chooses quantity $q_{-i}$ is given by
\[
l_i(q_i,q_{-i})=\sup_{P\in \mathcal P} \left(\sup_{q'_i\ge 0} u_i(q'_i,q_{-i}; P)-u_i(q_i,q_{-i}; P)\right).
\]
The maximum loss describes how much more profit firm $i$ could have obtained if it had known the inverse demand $P$ when anticipating that the other firm produces $q_{-i}$. Firm $i$'s {\it best compromise} given a choice $q^*_{-i}$ of the other firm is a quantity $q^*_i$ that achieves the lowest maximum loss, so
\[
q^*_i\in\argmin_{q_i\ge 0} l_i(q_i,q_{-i}).
\]
A strategy profile $(q^*_1,q^*_2)$ is a {\it perfect compromise equilibrium} if each firm chooses a best compromise given the choice of the other firm.

\begin{proposition}\label{p:cournot}
There exists a unique perfect compromise equilibrium. In this PCE, the strategy profile $(q^*_1,q^*_2)$  is given by
\beq\label{e-cournot-eq}
q_i^*=\frac{1}{3 \left(\sqrt{\ul b}+\sqrt{\ol b}\right)}\left(\frac{\ul a}{\sqrt{\ul b}}+\frac{\ol a}{\sqrt{\ol b}}\right), \ \ i=1,2.
\eeq
The associated maximum losses are
\beq\label{e-cournot-com}
l_i(q^*_i,q^*_{-i})=\frac{(\ul a\ol b-\ol a\ul b)^2}{4 \ul b\ol b \left(\sqrt{\ul b}+\sqrt{\ol b}\right)^2}, \ \ i=1,2.
\eeq
\end{proposition}

The proof is in Appendix \ref{s:p1}.

Let us discuss the strategic concerns underlying the PCE in this game. Each firm $i$, when facing unknown inverse demand and deciding about the quantity to produce, worries about two possible situations. It could be that the inverse demand is actually very high, so the firm is losing profit by producing too little. The greatest such loss occurs when the inverse demand is the highest, so $ P=\ol P$. Alternatively, it could be that the inverse demand is actually very low, so the firm is losing profit by producing too much. The greatest such loss occurs when the inverse demand is the lowest, so $ P=\ul P$. The firm thus chooses the best compromise $q_i^*$ that balances these two losses, assuming that the other firm follows its equilibrium strategy $q_{-i}^*$.

\begin{remark}
It is generally intractable to find a PBE in this game with such a rich set of possible inverse demand functions. It can only be done under very specific priors about the inverse demand. For example, PBE can be found if a prior describes the uncertainty about the parameters of the linear inverse demand function $ P(q)=a-bq$ \citep{Vives}. 
\end{remark}

\begin{remark}
Our equilibrium analysis can shed light on how the firms' behavior changes in response to increasing uncertainty. For comparative statics, let us consider as a benchmark a linear inverse demand function $P_0(q)=a_0-b_0q$. We normalize constants $a_0$ and $b_0$ so that the monopoly profit is equal to 1, that is,
\[
\sup_{q\ge 0} (a_0-b_0 q)q=\frac{a^2_0}{4b_0}=1.
\]
Suppose that there is a small uncertainty. Specifically, for $\eps>0$ let $P(q)$ satisfy \eqref{e:demand} where
\[
\ul P(q)=\left(1-\frac{\eps}2\right)a_0-\left(1+\frac{\eps}2\right)b_0q \quad\text{and}\quad\ol P(q)=\left(1+\frac{\eps}2\right)a_0-\left(1-\frac{\eps}2\right)b_0q.
\]
Denote by $q^\eps=(q^\eps_1,q^\eps_2)$ the strategies of the PCE as given by Proposition \ref{p:cournot}.
We then obtain
 \[
\frac{\df q^\eps_i}{\df \eps} =\frac{2\eps}{3a_0}+O(\eps^3)>0.
 \]
So the firms optimally respond to a growing uncertainty about the demand by increasing their output, and do so at an increasing rate as $\eps$ grows. Next, consider the associated maximum losses as shown in \eqref{e-cournot-com}. Then
\[
l_i(q^\eps_i,q^\eps_{-i})=\eps^2+O(\eps^4), \ \ i=1,2.
\]
So the maximum losses in the PCE increase very slowly as uncertainty increases.
Moreover, if $\eps=0.1$, then $l_i(q^\eps_i,q^\eps_{-i})\approx 0.01$. So the firms lose no more than about 1\% of the maximum profit due to not knowing the demand.
\end{remark}

\subsection{Bertrand Duopoly with Private Costs}
We now consider how two firms compete in prices when the cost of the rival firm is unknown.

There are two firms that produce a homogeneous good. Each firm $i=1,2$ chooses a price $p_i$. Choices are made simultaneously. The consumers only buy from the firm that offers a lower price. In particular, the quantity that firm $i$ sells is given by
\[
q_i(p_i,p_{-i})=\begin{cases}
Q(p_i), & \text{if $p_i<p_{-i}$},\\
Q(p_i)/2, & \text{if $p_i=p_{-i}$},\\
0, & \text{if $p_i>p_{-i}$},
\end{cases}
\]
where $Q(p)$ is the demand function. For clarity of exposition we assume that the demand function is given by
\[
Q(p)=\max\left\{\frac{a-p}b,0\right\}
\] 
The cost of producing $q_i$ units is $ c_i q_i$. Each firm $i$'s profit is given by
\[
u_i(p_i,p_{-i}; c_i)=(p_i- c_i)q_i(p_i,p_{-i}),  \ \ i=1,2.
\]

Each firm $i$ knows her own marginal cost but not that of the other firm, and it is common knowledge that
\[
c_1, c_2\in [\ul c,\ol c], \ \ \text{where $0\le \ul c\le\ol c\le a/2$.}
\]
A firm $i$'s pricing strategy $s_i(c_i)$ describes its choice of the price given its marginal cost $c_i$.

For each marginal cost $c_i$, firm $i$'s {\it maximum loss} of choosing a price $p_i$ when facing pricing strategy $s_{-i}$ of the other firm is given by
\[
l_i(p_i,s_{-i};c_i)=\sup_{c_{-i}\in [\ul c,\ol c]} \left(\sup_{p'_i\ge 0} u_i(p'_i,s_{-i}(c_{-i}); c_i)-u_i(p_i,s_{-i}(c_{-i}); c_i)\right).
\]
The maximum loss describes how much more profit $i$ could have obtained if it had known the other firm's marginal cost $c_{-i}$, anticipating the other firm to follow the pricing strategy $s_{-i}$. Firm $i$'s {\it best compromise} given $c_i$ is a pricing strategy $s^*_i(c_i)$ that achieves the lowest maximum loss for a given strategy $s^*_{-i}$ of the other firm:
\[
s^*_i(c_i)\in\argmin_{p_i\ge 0} l_i(p_i,s^*_{-i};c_i).
\]
A strategy profile $(s^*_1,s^*_2)$ is a {\it perfect compromise equilibrium} if each firm $i$ chooses a best compromise given its marginal cost $c_i$ when facing the strategy $s^*_{-i}$ of the other firm.

\begin{proposition}\label{p:bertrand}
There exists a unique perfect compromise equilibrium. In this PCE, the pricing strategies are given by
\beq\label{e-bertrand-eq}
s_i^*( c_i)=\frac{1}{2}\left(a+ c_i-\sqrt{(a-\ol c)^2+(\ol c- c_i)^2}\right), \ \ i=1,2.
\eeq
The associated maximum losses are
\beq\label{e-bertrand-com}
l_i(s^*_i(c_i),s^*_{-i},c_i)=\frac{(a-\ol c)(\ol c-c_i)}{2}\le \frac{(a-\ol c)(\ol c-\ul c)}{2}, \ \ i=1,2.
\eeq
\end{proposition}

The proof is in Appendix \ref{s:p2}.

Let us discuss the strategic concerns underlying the PCE in this game. Each firm $i$, when deciding about the price $p_i>c_i$ and facing an unknown cost of the other firm, worries about two possible situations. It could be that the other firm chooses a weakly lower price $p_{-i}\le p_i$. Thus, firm $i$ could have obtained more profit by undercutting $p_{-i}$. The greatest such loss occurs when the other firm's price marginally undercuts $p_i$. Alternatively, it could be that the other firm chooses a higher price, $p_{-i}>p_i$. Thus, unless $p_i$ is the profit maximizing price for the monopoly, firm $i$ is losing profit by charging too little. The greatest such loss occurs when the other firm's cost is the highest possible, $\ol c$. The firm thus chooses the best compromise $s_i^*( c_i)$ that balances these two losses, assuming that the other firm follows its equilibrium strategy.

We find that the PCE price $s_i^*( c_i)$ is strictly increasing in $c_i$ and lies strictly above the marginal cost $c_i$ whenever $c_i<\ol c$. Moreover, $s_i^*(\ol  c)=\ol c$. So, any sale with the cost below $\ol c$ leads to a positive profit. The fact that the equilibrium price cannot not lie above $\ol c$ is intuitive. It is common knowledge that the costs are at most $\ol c$. So if a firm charges a price above $\ol c$, the other firm would undercut it. Note also that the largest equilibrium price cannot lie below $\ol c$. This is because a firm with cost $\ol c$ will never charge a price below $\ol c$.

Note that the lowest equilibrium price $s_i^*(\ul  c)$ is strictly positive, even if $\ul c=0$. This is because when the price is very low,
then the potential loss due to not undercutting the other firm is small, while the potential loss due to not setting a price much higher is large. This has an upward effect on prices.

\begin{remark}
It is generally intractable to find a PBE in this application under any reasonable prior, even in this simplest setting with linear demand and constant marginal costs. The PBE strategy profile for this simplest setting is implicitly defined by a differential equation with no closed form solution \citep[see][]{Spulber1995bertrand}. 
\end{remark}

\begin{remark}
As in Section \ref{s:cournot}, our equilibrium analysis can shed light on how the firms' behavior changes in response to increasing uncertainty. For comparative statics, let us consider as a benchmark marginal cost $ c_0=a/4$ (recall that we require $0\le  c_i\le a/2$, so $ c_0=a/4$ is the midpoint). We normalize the constants $a$ and $b$ of the demand function $Q(p)=(a-p)/b$ so that the monopoly profit is equal to 1, that is,
\[
\sup_{p\ge 0} (p- c_0)\frac{a-p}b=\frac{(a- c_0)^2}{4b}=1.
\]
Suppose that there is a small uncertainty. Specifically, for $0<\eps<1$ let $ c_i\in[\ul c,\ol c]$, $i=1,2$, where
\[
\ul c=\left(1-\frac{\eps}2\right) c_0\quad\text{and}\quad\ol c=\left(1+\frac{\eps}2\right) c_0.
\]
Denote by $s^\eps=(s^\eps_1,s^\eps_2)$ the PCE strategy profile as given by Proposition \ref{p:bertrand}.
We then obtain
 \[
\frac{\df s^\eps_i(c_i)}{\df \eps} =\frac{(a+c_i-2\ol c)c_0}{4\sqrt{(a-\ol c)^2+(\ol c-c_i)^2}}>0,
 \]
because, using our assumptions on the parameters,
\[
a+c_i-2\ol c\ge a-2\ol c=1-2\left(1+\frac{\eps}2\right) c_0=\frac 1 4(2-\eps)>0.
\]
So the firms optimally respond to the growing uncertainty about the demand by increasing their prices. They become less competitive. Next, consider the associated maximum losses as shown in \eqref{e-bertrand-com}. Then
\[
l_i(s^\eps_i(c_i),s^\eps_{-i},c_i)\le \frac {3\eps}{32}-\frac{\eps^2}{64}, \ \ i=1,2.
\]
So the maximum losses are small. For example, if $\eps=0.1$, then the maximum losses are bounded by $0.01$. So the firms lose no more than about 1\% of the maximum profit due to not knowing the cost of the other firm.
\end{remark}

\subsection{Public Good Provision}\label{s:public}
Here we consider a problem of public good provision where the public good is funded by contributions of its beneficiaries. We compare three different mechanisms that regulate the beneficiaries' payments. For each of these mechanisms, we investigate how the beneficiaries find compromises about how much to contribute towards the public good provision. Note that in this setting the Vickrey-Clarke-Groves mechanism is not feasible, because it requires the public good to be externally subsidized.

There are $n$ agents, each has a private value $v_{i}\in [0,\bar v]$ for a public good. Agents know their own values of the good, but not those of the others.
Each agent $i$ chooses how much to contribute for the public good provision. Let  $x_{i}\in [0,\bar v]$ be agent $i$'s contribution. The agents make their choices simultaneously.

A commonly known cost of providing the public good is $c>0$. To avoid considering multiple cases, we assume that this cost is not too high, specifically,
\begin{equation}\label{E:A1}
\frac{c}{n-1}\le \frac{\bar v}2.
\end{equation}

The payoffs are as follows. If the sum of the contributions does not cover the cost, so $\sum_{i=1}^{n}x_{i}< c$, then the public good is not provided and the agents' contributions are returned to them. In this case each agent $i$ obtains zero payoff. Otherwise, if $\sum_{i=1}^{n}x_{i}\ge c$, then the public good is provided, and each agent obtains the value of the good net of the contribution. In addition, the agents may be refunded the excess contribution, $\sum_{i=1}^{n}x_{i}-c$. The payoff of each agent $i$ is
\[
v_i-x_i+r_i(x),
\]
where $r_i(x)$ is a refund to agent $i$ that depends on the profile of contributions $x=(x_1,...,x_n)$. 
For all $x$ such that $\sum_{i=1}^n x_i\ge c$, the refunds $r_i(x)$ must satisfy:

(a) $r_i(x)\ge 0$ for each $i$,  so agents do not pay more than their contributions;

(b) $\sum_{i=1}^n (x_i-r_i(x))\ge c$, so the net payments cover the cost of the public good;

(c) $r=(r_1,...,r_n)$ is symmetric, so the agents are treated ex-ante equally.

We compare three simple refund rules.

\noindent (i) {\it No-refunds rule.} The excess contribution is not refunded to the agents, so
\begin{equation}\label{t-simple}
r_i(x)=0, \ \ i=1,...,n.
 \end{equation}
 (ii) {\it Equal-split rule.} The excess contribution is divided equally among to the agents, so 
\begin{equation}\label{t-additive}
 r_i(x)=\frac 1 n\left(\sum\nolimits_{j=1}^n x_j-c\right), \ \ i=1,...,n.
 \end{equation}
(iii) {\it Proportional rule.} The excess contribution is divided proportionally to the agents' individual contributions, so 
\begin{equation}\label{t-prop}
r_i(x)=\left(1-\frac{c}{\sum_{j=1}^n x_j}\right)x_i, \ \ i=1,...,n.
 \end{equation}

Let $s_i(v_i)$ be a strategy of agent $i$, so $x_i=s_i(v_i)$ specifies the contribution of agent $i$ whose private value is $v_i$. We restrict attention to strategies that are symmetric and undominated. Specifically, we assume that
\beq
\text{$s_i(v)=s_j(v)$ and $s_i(v)\le v$ for all $v\in[0,\ol v]$ and all $i,j\in\{1,...,n\}$.}\label{e-A2}
\eeq
The assumption that the strategies are symmetric is substantive, as we rule out potential asymmetric equilibria. The assumption that the strategies are undominated is inconsequential for the results and introduced for notational convenience.

An agent $i$'s {\it maximum loss} of choosing contribution $x_i$ when the other agents choose a profile of contributions $s_{-i}(v_{-i})$  describes how much more payoff agent $i$ could have obtained if she had known the true values of everybody else, anticipating that they follow their strategies. To determine the maximum loss, observe that agent $i$ worries about two possible situations. It could be that the total contribution is marginally below $c$, so  $x_i+\sum_{j\ne i} s_j(v_j)=c-\eps$ for a small $\eps>0$. The good is not provided, but had $i$ contributed $\eps$ more it would have been provided. As $\eps\to 0$, agent $i$'s loss is $v_i-x_i$. 
Alternatively, it could be that all other agents contribute enough to cover $c$, so $\sum_{j\ne i} s_j(v_j) \ge c$. Thus the agent could have contributed nothing and still received the good. In this case the loss is the amount of contribution net of the refund, $x_i-r_i(x_i,s_{-i}(v_{-i}))$. Agent $i$'s maximum loss is thus given by
\[
l_i(x_i,s_{-i};v_i)=\sup_{v_{-i}\in[0,\ol v]^{n-1}}\max\left\{v_i-x_i,x_i-r_i(x_i,s_{-i}(v_{-i}))\right\}.
\]
Agent $i$'s {\it best compromise} given $v_i$ is a strategy $s^*_i(v_i)$ that achieves the lowest maximum loss for a given strategy profile $s_{-i}$ of the other agents:
\[
s_i^*(v_i)\in \argmin_{x_i\in[0,v_i]} l_i(x_i,s_{-i};v_i).
\]
A strategy profile $s^*=(s_1^*,...,s_n^*)$ is a {\it perfect compromise equilibrium} if each agent $i$ chooses a best compromise given her value $v_i$ when facing the strategy profile $s^*_{-i}$ of the other agents.

In this application we are interested in how the agents' equilibrium behavior and total efficiency (welfare) changes in PCE induced by different refund rules.
We measure the efficiency of a strategy profile $s$ by the maximum welfare loss as compared to the complete information case. Because $s_i(v)\le v$ by assumption \eqref{e-A2}, the welfare loss only emerges in the case of $\sum_i s_i(v_i)<c\le \sum_i v_i$ where the good is not provided when it is efficient to do so. Our inefficiency measure is denoted by $L(s)$ and is given by
\beq\label{e:WLoss}
\begin{split}
&L(s)=\sup_{(v_1,...,v_n)\in[0,\bar v]^n}
\sum\nolimits_{i=1}^n v_i-c \\ 
&\text{subject to $\sum\nolimits_{i=1}^n s_i(v_i)<c\le  \sum\nolimits_{i=1}^n v_i$.}
\end{split}
\eeq

We now characterize the PCE and the associated welfare losses for each of the three refund rules.
 
 \begin{proposition}\label{C:PG}
For each of the three refund rules the is a unique PCE strategy profile $s^*=(s^*_1,...,s^*_n)$ that satisfies assumption \eqref{e-A2}. For each $i=1,...,n$ and each $v_i\in[0,\bar v]$,

(i) if $r_i(x)$ is the no-refunds rule, then
\[
s^*_i(v_i)=\frac{v_i}{2} \quad \text{and} \quad L(s^*)=c;
\]

(iii) if $r_i(x)$ is the equal-split rule, then
\[
s^*_i(v_i)=\frac{n}{2n-1}v_i \quad \text{and} \quad L(s^*)=\frac{n-1}{n}c;
\]

(iii)  if $r_i(x)$ is the proportional rule, then
\[
s^*_i(v_i)=\frac{v_i}{2}-c+\frac{1}{2}\sqrt{v_i^2+4 c^2} \quad \text{and} \quad L(s^*)=\frac{n}{n+1}c.
\]

\end{proposition}
The proof is in Appendix \ref{s:C:PG}.

Note that 
\[
c>\frac{n}{n+1}c>\frac{n-1}{n}c.
\] 
So, the equal split rule is more efficient than the other two according to our efficiency measure. 
This raises a question of the optimal design. Among all feasible refund rules, which one is the most efficient? 

Also note that, unlike the equal-split rule, the proportional rule leads to equilibrium behavior that is independent of the number of agents. So, it is robust to the agents' knowledge of how many of them there are. This raises another question. Suppose that the agents are ambiguous not only about the others' values, but also about how many agents there are. How does this change their best compromises, and which refund rule is the most efficient in this case? 

We leave these questions for future research.

\subsection{Job Market Signaling}
Here we investigate Spence's job market signaling \citep{Spence73} when the worker's productivity and cost of education are unknown to the firms.

There is a single worker and two firms. The worker has productivity $\theta$ with $\theta\in[0,1]$. The worker publicly chooses a level of education $e$, either low ($e_L$) or high ($e_H$), to signal her productivity to the firms. The cost of low education is zero. The cost of high education is $c$ with $c\ge 0$. The firms observe the worker's education level $e$ and simultaneously offer wages $w_1$ and $w_2$. The worker chooses the better of the two wages. Her payoff is given by
\[
v(w_1,w_{2},e;\theta,c)=\max\{w_1,w_2\}-\begin{cases}
0, & \text{if $e=e_L$},\\
c, & \text{if $e=e_H$}.
\end{cases}
\]
Each firm $i$'s payoff is given by
\[
u_i(w_i,w_{-i};\theta)=\begin{cases}
\theta-w_i, & \text{if $w_i>w_{-i}$},\\
(\theta-w_i)/2, & \text{if $w_i=w_{-i}$},\\
0, & \text{if $w_i<w_{-i}$}.
\end{cases}
\]

The worker knows her productivity type $\theta$ and her cost of high education $c$. The firms know neither. They only know that the worker can have any productivity $\theta$ in $[0,1]$ and that her cost of high education $c$ lies between two linearly decreasing functions of $\theta$. Specifically, $c$ is between $1-b\theta$ and $1-b\theta+\delta$, where $b$ and $\delta$ are parameters that satisfy $0\le \delta\le b\le 1$.
Formally, the firms know that $(\theta,c)$ belongs to the set $\Omega$ given by
\beq\label{e:spence-types}
\Omega=\left\{(\theta,c):
\theta\in[0,1] \ \text{and} \
c\in[1-b\theta, 1-b\theta+\delta].
\right\}
\eeq

The worker's strategy $e^*(\theta,c)$ describes her choice of the education level for each pair  $(\theta,c)\in\Omega$. Each firm $i$'s strategy $w^*_i(e)$ describes its wage offer conditional on each education level $e\in \{e_L,e_H\}$. 

Consider how a firm makes inference from the observed level of education of the worker. This is formalized with the notion of {\it speculated states}. Formally, these are the firms' degenerate beliefs that put probability one on specific states. Speculated states are the pairs $(\theta,c)$ that a firm thinks are possible after observing the education level of the worker. The set of speculated states is denoted by $S_i(e)$.
This set is {\it consistent} with the worker's equilibrium strategy $e^*$ if it includes all pairs $(\theta,c)$ under which the worker chooses $e\in\{e_L,e_H\}$, so $(\theta,c)\in S_i(e)$ if $e^*(\theta,c)=e$.

For each education level $e$, firm $i$'s {\it maximum loss} of choosing wage $w_i$ when the other firm chooses the wage according to its strategy $w^*_{-i}$ is given by
\[
l_i(w_i,w^*_{-i};e)=\sup_{(\theta,c)\in S_i(e)} \left(\sup_{w'_i\ge 0} u_i(w'_i,w^*_{-i}(e); \theta)-u_i(w_i,w^*_{-i}(e); \theta)\right).
\]
The maximum loss describes how much more profit firm $i$ could have obtained if it had known the true productivity and cost of education of the worker, anticipating that the other firm follows its strategy $w^*_{-i}$. Firm $i$'s {\it best compromise} given $e$ is a wage $w^*_i(e)$ that achieves the lowest maximum loss for a given strategy $w^*_{-i}$ of the other firm:
\beq\label{e-spence-s1}
w^*_i(e)\in\argmin_{w_i\ge 0} l_i(w_i,w^*_{-i};e).
\eeq
Observe that the worker has complete information. There is no need for a compromise. So, the worker simply chooses a best-response:
\beq\label{e-spence-s2}
e^*(\theta,c)\in\argmax_{e\in\{e_L,e_H\}} v(w^*_1(e),w^*_{2}(e),e;\theta,c).
\eeq

A profile $(e^*,w^*_1,w^*_2,S_1,S_2)$ of strategies and speculated states is a {\it perfect compromise equilibrium} (PCE) if two conditions hold. First, the strategies satisfy \eqref{e-spence-s1} and \eqref{e-spence-s2}, so each firm $i$ chooses a best compromise, and the worker chooses a best response to the strategies of the others. Second, the firms' sets of speculated states are consistent with the worker's strategy $e^*$.

A PCE is {\it pooling} if the worker chooses the same level of education for all $(\theta,c)\in\Omega$. A PCE is {\it separating} if the set $\Omega$ can be partitioned into two subsets such that worker types belonging to the same subset choose the same level of education, but these levels differ between the two subsets.

\begin{proposition}\label{p:spence}
(i) There exists a pooling PCE in which the worker chooses low education, so
\[
e^*(\theta,c)=e_L \ \ \text{for all $(\theta,c)\in\Omega$},
\]
and the firms' wages are given by
\[
w^*_i(e_H)=w^*_i(e_L)=\frac{1}2, \ \ i=1,2.
\]
After each observed education level $e$, each firm $i$'s set of speculated states $S_i(e)$ contains all states.

\smallskip\noindent (ii) If $\delta\ge 2 b^2-b$, then a separating PCE does not exist.

\smallskip\noindent (iii) If $\delta<2 b^2-b$, then there exists a separating PCE in which the worker chooses high education  if and only if her cost $c$ is at most $ \frac1{2b}(b-\delta)$, so for all $(\theta,c)\in\Omega$
\[
e^*(\theta,c)=\begin{cases}
e_H, &\text{if $c\le \frac1{2b}(b-\delta)$},\\
e_L, &\text{if $c>\frac1{2b}(b-\delta)$},
\end{cases}
\]
and the firms' wages are given by
\beq\label{e-spence-w0}
w^*_i(e_H)=\frac{1}2+\frac{b+\delta}{4b^2} \ \ \text{and} \ \  w^*_i(e_L)=\frac{\delta}{2b}+\frac{b+\delta}{4b^2}, \ \ i=1,2.
\eeq
After each observed education level $e$, each firm $i$'s set of speculated states $S_i(e)$ contains each state $(\theta,c)\in\Omega$ that satisfies
\begin{align}\label{e-spence-pb}
&\theta\in\left[0,\frac{b+\delta}{2b^2}+\frac\delta b\right] \ \ \text{if $e=e_L$}, \quad\text{and} \quad \theta\in\left[\frac{b+\delta}{2b^2},1\right] \ \ \text{if $e=e_H$.}
\end{align}
\end{proposition}
The proof is in Appendix \ref{s:p3}.

Let us discuss the strategic concerns underlying these PCE. Each firm $i$, when facing unknown productivity of the worker and deciding about the wage offer $w_i$, worries about two possible situations. It could be that the productivity is high, so offering a wage that is marginally greater than that of the competitor would improve profit. The greatest such loss occurs when the productivity is the highest possible. Alternatively, it could be that the productivity is low, so offering a wage that is smaller than the competitor's would eliminate the loss. The greatest such loss occurs when the productivity is the lowest possible. The firm thus offers the best compromise wage that balances these two losses, assuming that the other firm follows its equilibrium strategy. In equilibrium, both firms offer the same wage, so each of them has probability $1/2$ to hire the worker. This is the best compromise between not hiring a productive worker and hiring an unproductive worker at the specified wage.

An essential detail in the above considerations is that the greatest and smallest productivities are now endogenous and can depend on the level of education $e$ that the worker chooses. In the pooling equilibrium, $e=e_L$ does not provide any useful information, so all productivity types are possible. However, in the separating equilibrium, the firms believe that the productivity belongs to a different interval when observing a different level of education. For example, if $b=1$ and $\delta=1/4$, then the firms believe that $\theta\in[0,7/8]$ if the education is low, and that $\theta\in[5/8,1]$ if the education is high. 

Observe that, among the workers with productivity $\theta\in[5/8,7/8]$, some choose low education, while others choose high education. This overlap is due to the richness of the state space. The same productivity type $\theta$ can have different costs of education $c$ that can fall below or above the threshold at which high education is profitable. Clearly, this result cannot emerge in the traditional setting where the workers are differentiated only by their productivity.

The parameter $\delta$ captures the firms' uncertainty about the worker's cost of high education given her productivity type. As $\delta$ goes up, this range of costs increases. When $\delta$ is sufficiently large, education signaling is not very informative. A costly signal cannot be used to differentiate high and low productivity types, and the separating PCE does not exist.

\subsection{Bilateral Trade with Common Value}\label{s:trade}
We now examine bilateral trade with common value. In this example we show that trade can occur when traders follow a PCE. This is in stark contrast to the no-trade theorem under common values as predicted by PBE \citep{MilgromStokey82}.

A seller wants to sell an indivisible good to a buyer. The value $v$ of the good is the same for each of them. If the good is traded at some price $p$, then the buyer obtains $v-p$ and the seller obtains $p-v$. If the good is not traded, then both traders obtain zero.\footnote{The same analysis applies if the seller obtains $p$ when the good is sold and $v$ when the good is not sold.}

Neither trader knows $v$. Before the trade takes place, the traders privately consults independent experts to obtain some information about $v$. 
Each expert provides an interval of possible values, from the most pessimistic to the most optimistic assessment of the true value. Specifically, the seller privately learns that $v\in[x_0,x_1]$ and the buyer privately learns that $v\in[y_0,y_1]$. 

The traders commonly know the lower and upper bounds of the value $v$. These bounds are normalized to be 0 and 1, so $v\in[0,1]$. In addition, the traders commonly know that the experts cannot be wrong, so 
\begin{equation}\label{BT-cons}
v\in[x_0,x_1]\cap[y_0,y_1].
\end{equation}
We do not impose constraints on how precise or imprecise the experts' information is. We allow $[x_0,x_1]$ and $[y_0,y_1]$ to be arbitrary intervals contained in $[0,1]$ that satisfy \eqref{BT-cons}.

We consider a take-it-or-leave-it protocol in which the seller is the proposer. The protocol is as follows. First, the traders observe their private information $[x_0,x_1]$ and $[y_0,y_1]$. Then the seller asks a price $p\in [0,1]$. Finally, the buyer decides whether to accept or to reject the seller's asked price.

Let us describe the traders' strategies. Let $p^*(x_0,x_1)$ be the seller's asked price given her information $[x_0,x_1]$. Let $\alpha^*(p,y_0,y_1)$ be the buyer's decision whether to accept or to reject the asked price $p$ given the buyer's private information $[y_0,y_1]$, where  $\alpha^*(p,y_0,y_1)=1$ means to buy, and  $\alpha^*(p,y_0,y_1)=0$ means not to buy.

Next we describe how the buyer makes inference from the price asked by the seller. This is formalized with the concept of speculated values. These are values for $v$ that the buyer thinks are possible after he observes the price asked by the seller.
Let $V_b(p,y_0,y_1)$ be the buyer's set of speculated values when the seller asks price $p$. Clearly, the buyer rules out the values outside of $[y_0,y_1]$, so $V_b(p,y_0,y_1)\subset[y_0,y_1]$. But some values in $[y_0,y_1]$ may be ruled out too, because $p=p^*(x_0,x_1)$ depends on $x_0$ and $x_1$, and the buyer knows that $v\in[x_0,x_1]\cap[y_0,y_1]$. 

The buyer's maximum loss from his choice $\alpha\in\{0,1\}$, given the asked price $p$ and his set of speculated values $V_b(p,y_0,y_1)$, is
 \[
 l_b(\alpha;p,y_0,y_1)=\sup_{v\in V_b(p,y_0,y_1)}\big(\max\left\{v-p,0\right\}-\left(v-p\right)\alpha\big).
 \]
It describes how much more the buyer could have obtained if he knew the true value $v$. The seller's maximum loss of asking price $p$, given the buyer's acceptance strategy $\alpha^*$, is
 \[
 l_s(p;x_0,x_1)=\sup_{\substack{(v,y_0,y_1)\in[0,1]^3:\\ v\in[x_0,x_1]\cap[y_0,y_1]}}\left(\sup_{p'\in [0,1]} \left(p'-v\right)\alpha^*(p',y_0,y_1)-\left(p-v\right)\alpha^*(p,y_0,y_1) \right).
 \]
It describes how much more the seller could have obtained if she knew both $v$ and the buyer's private information $[y_0,y_1]$, anticipating that the buyer would follow his strategy $\alpha^*$. Each trader's {\it best compromise} is a choice that achieves the lowest maximum loss for a given strategy of the other trader. 
A strategy profile $(p^*,\alpha^*)$ is a {\it perfect compromise equilibrium} (PCE) if each trader chooses a best compromise given the strategy of the other trader.

\begin{proposition}\label{p:trade-b}
A perfect compromise equilibrium is given as follows. The seller asks
\beq\label{e-p4-1}
p^*(x_0,x_1)=\max\left\{\frac{x_0+x_1}{2}+ \frac{1-x_1}{4},\frac 1 2\right\}.
\eeq
If the seller asks $p\ge \tfrac 1 2$, then the buyer speculates that $v\in [\max\{y_0,2p-1\},y_1]$ and accepts this price if and only if
\[
p\le \frac{ y_0+y_1}{2}.
\] 
If the seller asks $p<\tfrac 1 2$, then the buyer speculates that $v\in \{y_0\}$ and accepts this price if and only if $p\le y_0$.
\end{proposition}

The formal proof is in Appendix \ref{s:p4}. 

Let us discuss the strategic concerns underlying this PCE.  
Consider first how the buyer makes his choice when the seller asks $p$. To build the intuition, let us first assume that the buyer makes no inference from the value of the asked price. So the buyer speculates that $v\in[y_0,y_1]$ and compares her maximal losses when buying and not buying the good. The maximal loss of buying is attained when $v=y_0$, giving the loss of $p-y_0$. The maximal loss of not buying is attained when $v=y_1$, giving the loss of $y_1-p$. The best compromise between these two situations is for the buyer to buy if and only if $p\le (y_0+y_1)/2$. 

Now consider the inference about $v$ that the buyer makes from the asked price $p$. When $p\ge 1/2$, the buyer concludes that $v$ cannot be below $2p-1$. This weakly increases the lower bound on $v$ to $\max\{y_0,2p-1\}$. When $2p-1\le y_0$, the inference from observing $p$ is not useful. So the buyer behaves as described above. When $2p-1>y_0$, the maximal loss from buying is larger than that from not buying. So the buyer does not buy. Notice that $2p-1>y_0$ implies $p>(y_0+y_1)/2$, and hence the rule described above continues to apply. In summary, the buyer behaves as if she ignores how the seller chooses the price when $p\ge 1/2$.

Alternatively, suppose that  $p<1/2$. This cannot happen in equilibrium, so the buyer can have any beliefs. Assume that the buyer speculates that $x_0=x_1=y_0$. So, the buyer speculates that $v=y_0$. Clearly, it is then best to buy the good if and only if $p\le y_0$.

Consider now how the seller chooses the price when anticipating the buyer's equilibrium behavior. Observe that $p$ should be at least $1/2$. This is because if $p<1/2$, then the buyer accepts $p$ if and only if $p\le y_0$. So the good will be purchased at $p<1/2$ only if it its value is above its price. Thus, choosing $p<1/2$ is dominated by not selling the good at all, which is achieved by choosing $p=1$.

To understand how a price $p\ge 1/2$ should be chosen, consider briefly an alternative setting where it is common knowledge that $v\in [x_0,x_1]=[y_0,y_1]$. So the buyer has the same information as the seller. Then the seller will ask $p=(x_0+x_1)/2$, as this is the highest price that the buyer is willing to accept, and any higher price leads to no sale with the same maximal loss.

Now return to our model. Assume that the buyer does not buy at price $p$. The seller's maximal loss is attained when the buyer would have bought at a marginally lower price, and moreover when the value of the good is the lowest, $v=x_0$. So the maximal loss equals $p-x_0$. Now assume that the buyer buys at price $p$. The seller's maximal loss is attained when the buyer is extremely optimistic and believes that $v\in [y_0,y_1]=[x_1,1]$. This buyer will accept any price up to $(x_1+1)/2$. So the maximal loss equals $(x_1+1)/2-p$. The seller chooses a best compromise price that balances these two losses, and hence sets $p=\frac 1 2(x_0+x_1)+ \frac 1 4 (1-x_1)$. Note that the price asked by the seller lies above the midpoint of the seller's interval $[x_0,x_1]$, due to the possibility of the extremely optimistic buyer.

Proposition \ref{p:trade-b} stands in contrast to the no-trade theorem under common values as predicted by PBE \citep{MilgromStokey82}. 
We observe that trade occurs in our PCE whenever the median assessment $\frac 1 2 (y_0+y_1)$ of the buyer exceeds  the price $p^*(x_0,x_1)$. The equilibrium price can be seen as an exaggeration of the seller's median assessment, because $p^*(x_0,x_1)>\frac 1 2 (x_0+x_1)$ unless $x_1=1$. 
The trade is possible because the traders cannot rule out the possibility of two opposing situations: winning and losing from trade. They do not want to miss a winning opportunity, but also they do not want to lose from trade. They compromise by choosing their decision thresholds so that they do not lose too much either way.

We hasten to point out that the PCE presented in Proposition \ref{p:trade-b} is not unique. For example, there is a no-trade PCE, where the seller always asks $p=1$, and the buyer accepts to buy the good at a price $p$ if and only if $p<y_0$. This equilibrium relies on a specific out-of-equilibrium belief of the buyer that $v=x_0=x_1=y_0$ whenever $p$ is different from 1. So, if the seller deviates to some price $p<1$, either the buyer rejects it, or the seller makes a loss.

\subsection{Forecasting}\label{s:forec}
We conclude the list of our examples with a forecasting problem. 
Here we consider a single agent who has to forecast of a random variable under multiple priors. This forecast is based on a noisy signal with a known distribution. In this example we illustrate how noise influences learning when the agent makes best compromise choices. 

In Appendix B we consider an alternative setting, where the agent knows the distribution of the random variable but she is ambiguous about the noisy signal. We also deal with the case where the agent is ambiguous about both aspects (Remark \ref{R:F}).

Consider an agent who has to forecast a random variable $\theta$ that belongs to $[0,1]$ and has a distribution $F$ with a density $f$.
The agent's payoff is the quadratic loss given by
\[
u(a,\theta)=-(a-\theta)^2,
\]
where $a\in[0,1]$ denotes a forecast. 

The agent can condition her forecast on a noisy signal $z$ about $\theta$. The signal generating process is given by a conditional probability distribution $G_\eps(z|\theta)$ with a parameter $\eps\in[0,1]$ specified as follows. Signal $z$ reveals the true value $\theta$ with probability $1-\eps$ and is drawn uniformly from $[0,1]$ with probability $\eps$, so
\beq\label{e-g-def}
G_\eps(z|\theta)=\begin{cases}
\eps z, &\text{if $z<\theta$,}\\
1-\eps+\eps z, &\text{if $z\ge \theta$.}
\end{cases}
\eeq 

The agent is ambiguous about the distribution $F$ of $\theta$. She knows that $F$ has mean $\theta_0$ and admits a density $f$ such that $\delta\le f(\theta)\le 1/\delta$ for some $\delta\in(0,1)$. This assumption on the density excludes holes in the support and point masses. The parameter $\delta$ can be interpreted as a lower bound on the degree of dispersion of $\theta$. The set of such distributions is
\[
\mathcal F_{\delta}=\left\{F\in\Delta([0,1]):\E_F[\theta]=\theta_0 \ \text{and} \ \delta\le f(\theta)\le 1/\delta \ \text{for all $\theta\in[0,1]$}\right\}.
\]
For each distribution $F\in \mathcal F_{\delta}$ the agent forms a posterior belief about $\theta$ conditional on the signal $z$, leading to a set of beliefs given $z$. 

Let $\E_{F,G_{\eps}}[\cdot|z]$ denote the conditional mean of $\theta$ when the agent speculates that $\theta$ is distributed according to $F$.
The {\it maximum loss} of a forecast $a\in[0,1]$ given a signal $z\in[0,1]$ is
\[
l(a;z)=\sup_{F\in\mathcal  F_{\delta}} \left(\sup_{a'\in[0,1]} \E_{F,G_{\eps}}[-(a'-\theta)^2|z]-\E_{F,G_{\eps}}[-(a-\theta)^2|z]\right).
\] 
It describes how much more the agent could have obtained if he knew the distribution $F$. A {\it best compromise} is a forecast $a^*(z)$ that achieves the smallest maximum loss,
\[
a^*(z)\in\argmin_{a\in[0,1]}l(a;z).
\]

\begin{proposition}\label{p:for1}
The agent's best compromise is
\[
a^*(z)=(1-\lambda)z+\lambda\theta_0,
\]
where
\[
\lambda=\frac{\eps}{2}\left(\frac{\delta}{1-\eps(1-\delta)}+\frac{1}{\delta+\eps(1-\delta)}\right).
\]
\end{proposition}

The proof is in Appendix \ref{s:P-F}.

Let us present some intuition behind Proposition \ref{p:for1}. Due to the quadratic penalty of making inaccurate forecasts, the loss of a forecast is equal to its distance from the expected mean conditional on the signal. The forecaster is worried about two possible situations, namely, when this conditional mean is high and when it is low. Consequently, the best compromise involves a forecast at the midpoint of these two extreme conditional means. Solving for this midpoint yields the formulae given in the statement of the proposition. In particular, the best compromise forecast lies between the ex-ante mean $\theta_0$ and the signal $z$.

Note that the agent's best compromise forecast depends on the precision $\eps$ of her signal and on the degree of the dispersion $\delta$ of the variable of interest. We show how each of these two parameters independently influences the best compromise forecast.

Fix the degree of dispersion $\delta$. When the agent's signal is not very noisy, then her forecast is close to the signal. This is because $a^*$ is continuous in $\eps$ and $\lim_{\eps\to 0} a^*(z)=z$. When the signal is very noisy, then her prediction is close to the ex-ante mean, as $\lim_{\eps\to1} a^*(z)=\theta_0$. 

Now we fix the precision $\eps$ of the noise and vary the bound $\delta$ on the degree of dispersion of $\theta$. As we relax the constraints on $F$ imposed by $\delta$, we obtain that the forecast approximates the midpoint between $\theta_0$ and $z$. Formally, $\lim_{\delta\to 0}a^*(z)=(\theta_0+z)/2$. This is because the best compromise balances two extreme situations. It could be that $F$ has very high dispersion, thus making the signal extremely valuable. On the other hand, it could that $F$ has very low dispersion, in which case the signal has very little value. The forecast seeks a best compromise between these two situations and selects the midpoint. 

Note that the above analysis and discussion reveals a discontinuity in the forecast $a^*$ at $\eps=\delta=0$.

\section{Conclusion}\label{s:concl}

We introduce a formal methodology to better understand how players deal
with uncertainty in dynamic strategic contexts. We are particularly
interested in modeling players who have an intuitive understanding of
uncertainty that can be expressed in terms of bounds. The general
setting looks at players who have ambiguous preferences that are modeled
as multiple priors. Learning occurs by updating prior by prior using
Bayes' rule whenever possible. Decisions are made under ambiguity by
finding best compromises.

Our objective is to present a solution concept that is as close as
possible to perfect Bayesian Equilibrium. The idea is to facilitate the
understanding and acceptance of PCE and simplify the interpretation of
new insights. This design objective also allows us to build on the
discipline underlying the concept of a PBE.

We identify at least six reasons that motivated us to create this new
solution concept, each of them motivated by contexts where PBE is not
adequate. These reasons are robustness, ambiguity, non-probabilistic
reasoning, parsimony, tractability, and accessibility. We explain each of
these in more detail.

{\it Robustness.} The PCE can be used to investigate the robustness of a PBE
to the priors of the players in a context where each of the players
seeks a strategy that also performs well for very similar priors.
Similarly it can be used to analyze how the play changes for a given PBE
when players only have an approximate understanding of some game elements.

{\it Ambiguity.} Ambiguous preferences have become popular. Our concept
allows us to include players with such preferences. The formalism we
introduce is not limited to the use of best compromises as the solution
concept. We could have also inserted any alternative concept for
decision making under ambiguity. The most prominent alternative is
maximin utility preferences that leads to a pessimistic mindset. We prefer the
flavor of finding compromises. Compromises seems necessary in a globalizing world where decision making is made in front of growing audiences and when there is less willingness
to base decisions on specific distributional assumptions.

{\it Non-probabilistic reasoning.} Uncertainty per se seems to mean that
details are hard to describe. And yet traditional models focus on two
types of workers, high and low, or capture the uncertainty by a small number of parameters.
Uncertainty seems to preclude that players agree on likelihoods of
events, and yet this is done in PBE. We introduce PCE to open the door
to understanding more realistic uncertainty.

{\it Parsimony.} The traditional PBE framework reveals a different solution
for each prior. Such flexibility can be useful to fit data. But
flexibility in terms of a multitude of different answers gives little
guidance to those who need to make choices. One easily loses the big
picture if there are many details that determine what happens. To
achieve clear and transparent results, one often gives up realism and
adapts simplistic uncertainty with only a few types for each player. In
contrast, the PCE concept under genuine ambiguity is by design very
parsimonious. Making best compromises across many different situations
allows to abstract from many details.

{\it Tractability.} The usefulness of our solution concept is demonstrated in
relevant economic examples where uncertainty is rich. This richness
limits a tractable analysis of PBE. PCE yields tractable results with
simple proofs as players focus on extreme situations, allowing them to
ignore intermediate constellations.

{\it Accessibility.} The PCE concept under genuine ambiguity is undemanding and easy to teach. Uncertainty can be described with bounds. There is no need for probabilities, and Bayes' rule can be put back on the shelf. 

The common acceptance of priors is dwindling. The literature on decision
making and game playing under uncertainty has now developed alternative
concepts. We hope to add to this literature. Numerous paths to future
research open up in a search for new insights and for a clearer
exposition of existing understanding of economic and strategic principles.

\section*{Appendix A. Proofs.}
\renewcommand{\thesection}{A}

\subsection{Proof of Theorem \ref{p:exist}}\label{s:pt}
Consider a game $\Gamma=(N,\mathcal G,\Omega,(\Pi_1,...,\Pi_n),(u_1,...,$ $u_n))$.  Let $\Phi$ be the set of information sets excluding the initial node $\phi_0$, so $\Phi=\bigcup\nolimits_{i\in N} \Phi_i$. Recall that $A_\phi$ is the set of pure actions of the player who moves at information set $\phi\in\Phi$. A strategy profile $s$ associates with each information set $\phi$ a mixed action $s_\phi\in\mathscr A_{\phi}=\Delta(A_\phi)$ at $\phi$. 

We now define an $\eps$-perturbed game. Let $\eps$ be a small enough positive number. Let $\Delta_\eps(A(\phi))$ be the set of mixed actions at information set $\phi$ such that each pure action in $A(\phi)$ is played with probability at least $\eps$. Let $\mathcal S_\eps$ be the set of strategy profiles such that $s_\phi\in\Delta_\eps(A(\phi))$ for each $\phi\in\Phi$. So the strategies in $\mathcal S_\eps$ are completely mixed. An {\it $\eps$-perturbed game $\Gamma_\eps$} is the original game $\Gamma$ where the players' strategies are confined to $\mathcal S_\eps$.

Consider a strategy profile $s\in\mathcal S_\eps$. Because $s$ is fully mixed, the belief system that is consistent with $s$ is uniquely defined by Bayes' rule. Denote this belief system by $\beta(s)$, and let $\beta_{\phi}(\pi;s)$ is the posterior probability distribution over the decision nodes in the information set $\phi$ derived from a prior $\pi$. Let
\[
B_{\phi_i}(s)=\left\{\beta_{\phi_i}(\pi_i;s):\pi_i\in\Pi_{i}\right\}
\]
be the set of beliefs at each $\phi_i\in\Phi_i$ for each player $i\in N$.
Let $U_{\phi_i}(s)$ be the negative of player $i$'s maximum loss at $\phi_i\in\Phi_i$ when player $i$ follows her strategy $s_{\phi_i}$, so
\begin{align}
U_{\phi_i}(s)&=-l(s_{\phi_i}|s,\beta(s),\phi)\notag\\
&=\inf_{b_i\in B_{\phi_i}(s)}\left(\bar u_i(s_{\phi_i}|s,\phi_i,b_i)- \sup_{a_i\in A(\phi_i)} \bar u_i(a_i|s,\phi_i,b_i)\right).\label{e-U-phi}
\end{align}
Two observations are in order. First, $U_{\phi_i}(s)=U_{\phi_i}(s_{\phi_i},s_{-\phi_i})$ is continuous in $s_{\phi_i}$. This is because $\ol u$ is continuous, and the set $B_{\phi_i}(s)$ of beliefs at $\phi_i$ is independent of $s_{\phi_i}$ (it only depends on the choices in the information sets preceding $\phi_i$). Second, $U_{\phi_i}(s_{\phi_i},s_{-\phi_i})$ is also continuous in $s_{-\phi_i}$ when $s\in S_\eps$, so the strategies are fully mixed. This is because $B_{\phi_i}(s)$ is a continuous correspondence w.r.t.~$s\in S_\eps$, as it is derived by Bayes' rule from the set of priors pointwise, and Bayes' rule is a well defined and continuous operator for $s\in S_\eps$. In addition, both $B_{\phi_i}(s)$ and $A(\phi_i)$ are compact. The continuity of $U_{\phi_i}(s_{\phi_i},s_{-\phi_i})$ in $s_{-\phi_i}$ then follows from the Maximum Theorem \citep{Berge}.

We now construct an augmented game $(\Phi,\mathcal G,\Omega,\pi_0,U)$ as follows. Let each information set $\phi\in\Phi$ be associated with a different player, so the set of players is the set of information sets $\Phi$. The game tree  $\mathcal G$ and the set of states $\Omega$ remain unchanged. Let $\pi_0$ be a common prior over the states, and assume that $\pi_0$ has full support over $\Omega$. Nature moves first by choosing a state $\omega\in\Omega$ according to the prior $\pi_0$. Each player $\phi\in\Phi$ moves only once, at her information set $\phi$, by choosing a mixed action from the set $\Delta_\eps(A(\phi))$. The interim payoff of each player $\phi\in\Phi$ at the information set $\phi$ is given by $U_\phi(s)$. 
Let $U=(U_\phi)_{\phi\in\Phi}$. 

The augmented game $(\Phi,\mathcal G,\Omega,\pi_0,U)$ can be seen as a game of incomplete information with a nonstandard specification of the players' payoffs. While in a standard game the payoffs are specified ex-post at each terminal node, in this augmented game the payoff $U_\phi$ of each player $\phi\in\Phi$ is specified in the interim, at the information set where the player makes a move. Because each player moves only once, the specification of the interim payoffs is sufficient to apply the concept of PBE or sequential equilibrium to the augmented game. 

Another nonstandard feature of the augmented game is that each player's interim payoff $U_\phi(s)$ depends on the set of beliefs $B_\phi(s)$ at $\phi$, but it is independent of the state $\omega$ itself. So, the prior $\pi_0$ does not affect the best-response actions by the players, it only affects the likelihood of reaching different information sets in the game tree.

Let $(s'_{\phi},s_{-\phi})\in\mathcal S_\eps$ denote the strategy profile where $s'_{\phi}$ is played by player $\phi$ and $s_{-\phi}$ is the profile of strategies at all other players. Observe that maximizing $U_{\phi}(s'_\phi,s_{-\phi})$ with respect to player $\phi$'s own decision $s'_\phi\in \Delta_\eps(A(\phi))$ is the same as minimizing the maximum loss at $\phi$ in the perturbed game $\Gamma_\eps$. Consequently, if $\ol s$ is a strategy profile in a sequential equilibrium of the augmented game, then $(\ol s,\beta(\ol s))$ is a PCE of $\Gamma_\eps$. The existence of PCE follows from the existence of sequential equilibrium for finite games. We refer the reader to \cite{Chakrabarti2016} for the backward-induction proof of existence of sequential equilibrium that uses interim payoffs at information sets to determine players' best-response correspondences.

Thus we have shown the existence of a PCE in every perturbed game $\Gamma_\eps$. It remains to show the existence of a PCE in the original, unperturbed game $\Gamma$. Consider a sequence $(\eps_k)_{k=1}^\infty$ such that $\lim_{k\to\infty} \eps_k=0$. Let $(s^k,\beta^k)$ be a PCE for the perturbed game $\Gamma_{\eps_k}$. By Bolzano-Weierstrass theorem there exists a subsequence $(k_t)_{t=1}^\infty$ such that $(s^{k_t},\beta^{k_t})$ converges to some $(s^*,\beta^*)$ as $t\to\infty$.
Observe that the belief system $\beta^*$ is consistent with $s^*$. This is because for each player $i$, each information set $\phi_i\in\Phi_i$, and each prior $\pi_i\in\Pi_i$, either $\beta^*_{\phi_i}(\pi_i)$ is derived by Bayes rule that is continuous as $(s^{k_t},\beta^{k_t})$ approaches $(s^*,\beta^*)$, or Bayes rule is undefined in the limit, in which case $\beta^*_{\phi_i}(\pi_i)$ is also consistent by definition. Next, for all $\eps>0$, all $t$ such that $\eps\ge \eps_{k_t}$, and all $s'_\phi\in \Delta_\eps(A_{\phi})$ we have
\begin{align*}
0\le &U_{\phi}(s^{k_t}_{\phi},s^{k_t}_{-\phi})-U_{\phi}(s'_\phi,s^{k_t}_{-\phi})=-l(s^{k_t}_{\phi}|s^{k_t},\beta^{k_t},\phi)+l(s'_\phi|s^{k_t},\beta^{k_t},\phi)\xrightarrow{t\rightarrow\infty} \\
&-l(s^*_{\phi}|s^*,\beta^*,\phi)+ l(s'_\phi|s^*,\beta^*,\phi)=U_{\phi}(s^*_{\phi},s^*_{-\phi})-U_{\phi}(s'_\phi,s^*_{-\phi}),
\end{align*}
where the inequality is by $s^{k_t}_{\phi}$ being a best response in the augmented game, the first equality is by \eqref{e-U-phi}, the limit is by the continuity of $l(s_{\phi}|s,\beta,\phi)$ in $s$ and $\beta$, and the second equality is because the set $B_{\phi}(s^{k_t})$ of beliefs at $\phi$ is independent of the mixed action $s^{k_t}_\phi$ at $\phi$. It follows that $s^*_{\phi}$ is a best response to $s^*_{-\phi}$. So $s^*$ is a best compromise strategy profile in the unperturbed game $\Gamma$. We thus have shown that $(s^*,\beta^*)$ is a PCE of $\Gamma$.
\qed

\subsection{Proof of Proposition \ref{p:lemon}}\label{s:lemon-proof}
Each of the seller's two information sets (one for each type, low and high) contains a single decision node. Hence the seller's beliefs at these decision nodes are trivial, and the seller's best compromises are simply her best responses. In the high information set ($\theta=\theta_H$), choosing $\sigma^*_S(\theta_H)=1$ is the strictly dominant strategy. We now consider two possibilities of the seller's choice in the low information set: $\sigma^*_S(\theta_L)>0$ and $\sigma^*_S(\theta_L)=0$.

The buyer has a single information set $\phi_B$ and forms a set of belief $B_{\phi_B}$. Each belief $b\in B_{\phi_B}$ is given by $b=\beta_{\phi_B}(\pi_k)$, $k=0,1,...,K$ and denotes the probability of being in the decision node where the state is high, $\theta=\theta_H$.

First, suppose that $\sigma^*_S(\theta_L)>0$. Then, for each prior $\pi_k\in\Pi_B$, a buyer's belief $b\in B_{\phi_B}$ is consistent with $\sigma^*_S$ if it is given by Bayes' rule \eqref{e:Bayes}. In particular, $\beta_{\phi_B}^*(0)=0$ and $\beta_{\phi_B}^*(1)=1$. From \eqref{e:loss-lemon} it is evident that the maximum loss for the buyer is determined at the extreme beliefs, 0 and 1. Substituting these into \eqref{e:loss-lemon} yields
\[
l_S(\sigma_B|\phi_B,\sigma_S^*,\beta^*_{\phi_B})=\max\big\{(p-c_B)(1-\sigma_B),c_B\sigma_B\big\},
\]
which is minimized by $\sigma_B^*=(p-c_B)/p$. However, using assumption \eqref{e-assum-lemon}, we obtain
\[
\frac{p-c_B}{p}>\frac{p}{p+c_S}.
\]
Consequently, by \eqref{e-BR-S-lemon}, the unique best response of the low type seller to the buyer's strategy $\sigma_B^*=(p-c_B)/p$ is $\sigma^*_S(\theta_L)=0$, which contradicts the initially assumed $\sigma^*_S(\theta_L)>0$. We thus conclude that there is no PCE where $\sigma^*_S(\theta_L)>0$.

Alternatively, suppose that $\sigma^*_S(\theta_L)=0$. Then, for each prior $\pi_k\in\Pi_B$ with $\pi_k>0$,  Bayes' rule \eqref{e:Bayes} yields the belief $b=1$. However, for the prior $\pi_0=0$, Bayes' rule does not apply, so every belief $b_0\in [0,1]$ is consistent with $\sigma^*_S$. We thus obtain $B_{\phi_B}=\{b_0,1\}$ for some $b_0\in[0,1]$. Substituting these beliefs into \eqref{e:loss-lemon} yields
\[
l_S(\sigma_B|\phi_B,\sigma_S^*,\beta^*_{\phi_B})=\max\big\{((1-b_0)p-c_B)(1-\sigma_B),c_B\sigma_B\big\},
\]
which is minimized by
\[
\sigma^*_B=\begin{cases}
0, & \text{if $b_0\in[\frac{p-c_B}{p},1]$,}\\
\frac{(1-b_0)p-c}{(1-b_0)p}, &  \text{if $b_0\in[0,\frac{p-c_B}{p})$}.
\end{cases}
\]
By \eqref{e-BR-S-lemon}, if $\sigma^*_B<p/(p+c_S)$ (in particular, if $\sigma^*_B=0$), then the unique best response of the low type seller is $\sigma^*_S(\theta_L)=1$, which contradicts the initially assumed $\sigma^*_S(\theta_L)=0$. However, for each $b_0$ that satisfies
\beq\label{e:cond:sigma}
\sigma^*_B=\frac{(1-b_0)p-c}{(1-b_0)p}\in\left[\frac{p}{p+c_S},1\right].
\eeq
the strategy $\sigma^*_S(\theta_L)=0$ is a best response for the seller. Thus, $(\sigma^*_S,\sigma^*_B)$ with $b_0$ that satisfies \eqref{e:cond:sigma} is a PCE pair of strategies. Finally, observe that the interval of $b_0$ that satisfies \eqref{e:cond:sigma} is given by $\left[0,1-c_B/c_S-c_B/p\right]$.
\qed

\subsection{Proof of Proposition \ref{p:cournot}}\label{s:p1}
To prove the existence of a unique PCE, we find a unique profile of best-compromise strategies and a unique profile of beliefs that satisfy Definition \ref{def:consistency}.

First, we find the beliefs. The firms have genuine ambiguity, so the set of priors $\Pi_i$ of firm $i$ is equal to the set of degenerate beliefs over $\mathcal P$. By Definition \ref{def:consistency} and the consistency requirement in PCE, the set $B_i(\phi_i)$ of beliefs of firm $i$ at its unique information set $\phi_i$ must be equal to the set of priors, so $B_i(\phi_i)=\Pi_i$.

Next, we find each firm's equilibrium quantity. For derivations, we assume that the quantities and the price are always nonnegative, and then we verify  that this is indeed the case in equilibrium.

Let $x^*_i(q_{-i}, P)$ be a best response strategy of player $i$ given the knowledge of $q_{-i}$ and the inverse demand function $ P$. The loss of firm $i$ from choosing quantity $q_i$, given $q_{-i}$ and $P$, is denoted by $\Delta u_i(q_i,q_{-i}; P)$ and given by
\[
\Delta u_i(q_i,q_{-i}; P)= P(x_i^*(q_{-i}, P)+q_{-i})x_i^*(q_{-i}, P)- P(q_i+q_{-i})q_i.
\]
By \eqref{e:demand}, the marginal revenue of firm $i$ satisfies
\[
\ul  P(q_i+q_{-i})+\ul  P'(q_i+q_{-i})q_i\le  P(q_i+q_{-i})+ P'(q_i+q_{-i})q_i\le  \ol P(q_i+q_{-i})+\ol P'(q_i+q_{-i})q_i.
\]
Therefore, for given $q_{j}$ and $ P$, the best-response quantity $x^*_i(q_{-i}, P)$ of firm $i$ always lies between $x^*_i(q_{-i},\ul  P)$ and  $x^*_i(q_{-i},\ol  P)$. While the profit function need not be concave in general, it is concave when $ P=\ul  P$ or when $ P=\ol  P$. So the highest loss will always be attained in one of these two extreme cases:
\[
l_i(q_i,q_{-i})=\sup_{ P} \Delta u_i(q_i,q_{-i}; P) =\max\{ \Delta u_i(q_i,q_{-i};\ul P), \Delta u_i(q_i,q_{-i};\ol P)\}.
\]
It is easy to see that the maximum loss is minimized by balancing the two expressions under the maximum:
\[
\Delta u_i(q_i,q_{-i};\ol  P)=\Delta u_i(q_i, q_{-i};\ul  P).
\]
Substituting $\ul  P$ and $\ol  P$ and simplifying the expressions yields the equation
\beq\label{e-maxmin-cournot}
\frac{(\ol a-\ol b q_{-i})^2}{4\ol b}-(\ol a-\ol b(q_i+ q_{-i}))q_i=\frac{(\ul a-\ul b q_{-i})^2}{4\ul b}-(\ul a-\ul b(q_i+ q_{-i}))q_i.
\eeq
Solving for $q_i$ yields the unique best compromise quantity:
\[
q^*_i=\frac{\ul a\sqrt{\ol b}+\ol a\sqrt{\ul b}}{2(\ul b\sqrt{\ol b}+\ol b\sqrt{\ul b})}-\frac{q_{j}}2, \ \ i=1,2.
\]
Solving this pair of equations for $(q^*_1,q^*_2)$, we find \eqref{e-cournot-eq}.
It is easy to verify that  under our assumptions, $q^*_i>0$, and moreover, $ P(q^*_1+q^*_2)\ge\ul  P(q^*_1+q^*_2)>0$. Substituting the solution into \eqref {e-maxmin-cournot} yields the maximum loss of each firm \eqref{e-cournot-com}.
\qed

\subsection{Proof of Proposition \ref{p:bertrand}}\label{s:p2}
Similarly to the proof of Proposition \ref{s:p1}, to prove the existence of a unique PCE, we find a unique profile of best-compromise strategies and a unique profile of beliefs that satisfy Definition \ref{def:consistency}.

First, we find the beliefs. The firms have genuine ambiguity, so the set of priors $\Pi_i$ of firm $i$ is equal to the set of degenerate beliefs over $[\ul c,\ol c]^2$. By Definition \ref{def:consistency} and the consistency requirement in PCE, firm $i$ with cost $c_i$ must have the set $B_i(c_i)$ of beliefs equal to the set of priors, so $B_i(\phi_i)=\Pi_i$.

Next, we find each firm's equilibrium quantity. 
For derivations, we assume that each firm prices at or above marginal cost, and then we verify that this is indeed the case in equilibrium.

Consider firm $i$ with type $c_i\in[\ul c,\ol c]$. Let $s^m(c_i)$ be the profit-maximizing pricing strategy if firm $i$ were the monopoly, so $s^m(c_i)=(a+ c_i)/2$. Since we have assumed that $\ol c\le a/2$, this means that $s^m(c_i)\ge\ol c$ for all $c_i$. The monopoly profit is $(a-c_i)^2/(4b)$. 

Fix the other firm's strategy $s^*_{-i}(c_{-i})$ and let $\ol p$ be the maximum price of the other firm, so $\ol p=\sup\nolimits_{c_{-i}\in [\ul c,\ol c]} s^*_{-i}(c_{-i})$. Given the other firm's cost $c_{-i}$, and thus the price $p_{-i}=s^*_{-i}(c_{-i})$, firm $i$'s maximum profit is
\begin{align*}
u^*_i(p_{-i};c_i)=\sup_{x_i\ge 0}u_i(x_i, p_{-i};c_i)&= \begin{cases}
0, &\text{if $p_{-i}\le c_i$},\\
(p_{-i}-c_i)\frac{a-p_{-i}}{b}, &\text{if $c_i<p_{-i}\le s^m(c_i)$},\\
\frac{(a-c_i)^2}{4b}, &\text{if $p_{-i}> s^m(c_i)$}
\end{cases}\\
&=\max\left\{0,(p_{-i}-c_i)\frac{a-p_{-i}}{b},\frac{(a-c_i)^2}{4b}\right\}.
\end{align*}

Let $p_i$ be a price of firm $i$. We now find the maximum loss of firm $i$ from choosing $p_i$, given its marginal cost $c_i$ and the strategy $s^*_{-i}$ of the other firm. There are three cases. 

First, suppose that $p_{-i}\le c_i\le p_i$. Then firm $i$ cannot make positive profit, so $p_i$ is a best response. Thus, firm $i$ behaves optimally in this case, so the loss is zero.

Second, suppose that $c_i<p_{-i}\le p_i$. Then firm $i$ 
could have been better off by marginally undercutting $p_{-i}$. Maximizing the loss over $p_{-i}\in(c_i,p_i]$, we obtain
\beq\label{e-loss-1}
\sup_{p_{-i}\in(c_i,p_i)}\left(u^*_i(p_{-i};c_i)-u_i(p_i,p_{-i};c_i)\right)=
\begin{cases}
(p_{i}-c_i)\frac{a-p_{i}}{b}, &\text{if $p_{i}\le s^m(c_i)$,}\\
\frac{(a-c_i)^2}{4b}, &\text{if $p_{i}>s^m(c_i)$}.
\end{cases}
\eeq
Third, suppose that $p_i<p_{-i}$. Then firm $i$ could have made more profit by increasing its price, so its maximum loss is
\begin{multline}\label{e-loss-2}
\sup_{p_{-i}\in(p_i,\ol p]}\left(u^*_i(p_{-i};c_i)-u_i(p_i,p_{-i};c_i)\right)=u^*_i(\ol p_;c_i)-u_i(p_i,\ol p;c_i)\\
=-(p_{i}-c_i)\frac{a-p_{i}}{b}+\begin{cases}
(\ol p-c_i)\frac{a-\ol p}{b}, &\text{if $p_i\le s^m(c_i)$,}\\
\frac{(a-c_i)^2}{4b}, &\text{if $p_i>s^m(c_i)$}.
\end{cases}
\end{multline}
To minimize the maximum loss, we need to minimize the greater of the expressions in \eqref{e-loss-1} and \eqref{e-loss-2}. Observe that, by the definition of $s^m(c_i)$, the right-hand side in \eqref{e-loss-1} is constant and the right-hand side in \eqref{e-loss-2} is strictly increasing in $p_i$ for $p_i>s^m(c_i)$. So we only need to consider $p_i\le s^m(c_i)$. Under this assumption, the greater of 
 the expressions in \eqref{e-loss-1} and \eqref{e-loss-2} can be simplified to
\[
l_i(p_i,s^*_{-i};c_i)=\max\left\{(p_{i}-c_i)\frac{a-p_{i}}{b},(\ol p-c_i)\frac{a-\ol p}{b}-(p_{i}-c_i)\frac{a-p_{i}}{b}\right\}.
\]
Because one expression is increasing and the other is decreasing in $p_i$ for $p_i\le s^m(c_i)$, the maximum loss is minimized at the solution of
\beq\label{e-bertrand-loss}
(p_{i}-c_i)\frac{a-p_{i}}{b}=(\ol p-c_i)\frac{a-\ol p}{b}-(p_{i}-c_i)\frac{a-p_{i}}{b}.
\eeq
Solving the above for $p_i$ and assigning $s^*_i(c_i)=p_i$, we obtain \eqref{e-bertrand-eq}.

To see that $s_i^*(c_i)\ge c_i$, observe that
\[
s_i^*(c_i)-c_i=\frac{1}{2}\left(a-c_i-\sqrt{(a-\ol c)^2+(\ol c- c_i)^2}\right)\ge 0
\]
by the triangle inequality and $a>\ol c\ge c_i$. Moreover, $s_i^*(c_i)>c_i$ when $c_i<\ol c$, and $s^*_i(\ol c)=\ol c$. 
Finally, substituting $s^*_i(c_i)$ into the maximum loss expression in \eqref{e-bertrand-loss} yields \eqref{e-bertrand-com}.
\qed

\subsection{Proof of Proposition \ref{C:PG}}\label{s:C:PG}
We prove only part (iii) of Proposition \ref{C:PG} for the proportional rule given by \eqref{t-prop}. The proof of parts (i) and (ii) for the other two rules is analogous but easier, and thus omitted.

Let the refunds $r_i$ be given by the proportional rule \eqref{t-prop}. We first derive an agent $i$'s best compromise strategy $s^*_i$. Agent $i$ who chooses $x_i$ worries about two possible situations.
It could be that the total contribution is marginally below $c$, so  $x_i+\sum_{j\ne i} s_j(v_j)=c-\eps$ for a small $\eps>0$. The good is not provided, but had $i$ contributed $\eps$ more it would have been provided. As $\eps\to 0$, agent $i$'s loss is $v_i-x_i$. 
Alternatively, it could be that all other agents contribute enough to cover $c$, so $\sum_{j\ne i} s_j(v_j)\ge c$. Thus the agent could have contributed nothing and still received the good. In this case the loss is the amount of contribution net of the refund, $x_i-r_i(x)$. This loss is maximized when the other agents' contributions exactly equal to the cost, so $\sum_{j\ne i} s_j(v_j)=c$, so by \eqref{t-prop} we have
\[
x_i-r_i(x)=\frac{cx_i}{x_i+\sum_{j\ne i} s_j(v_j)}\le \frac{cx_i}{x_i+c}.
\]
The loss in the first case is weakly decreasing and the loss in the second case is strictly increasing in $x_i$. To find $x_i$ that minimizes the maximum loss, we solve the equation
\[
v_i-x_i=\frac{cx_i}{x_i+c}
\]
for $x_i$. Denote the solution by $s^*(v_i)$. It is easy to verify that it is as given in part (iii) of the statement of Proposition \ref{C:PG}. Note that it is symmetric across the players, so we drop the subscript $i$. 

The above argument requires that there exist values $v_j\in[0,\ol v]$ such that $\sum_{j\ne i} s^*(v_j)=c$. Observe that $s^*(0)=0$ and $s^*(v_i)$ is increasing in $v_i$. So, we only need to verify that $\sum_{j\ne i}s^*(\ol v)\ge c$, which holds under condition \eqref{E:A1}. 

It remains to determine the maximum welfare loss $L(s^*)$ as defined in \eqref{e:WLoss}. 
As $s^*(v_i)$ is increasing in $v_i$, the constraint $\sum\nolimits_{i=1}^n s^*(v_i)<c$ must be binding. Moreover, it is easy to verify that $s^*(v_i)$ is convex in $v_i$. Thus, by Jensen's inequality we have
\[
\sum\nolimits_{j=1}^n s^*(v_j)\ge n s^*\left(\frac{1}{n}\sum\nolimits_{j=1}^n v_j\right).
\]
Thus, the maximum is attained for $v_1=...=v_n=z$ for $z\in[0,\ol v]$ such that $n s^*(z)=c$. Solving the equation
\[
n\left(\frac{z}{2}-c+\frac{1}{2}\sqrt{z^2+4 c^2}\right)=c
\]
for $z$ yields
\[
z=\frac{2n+1}{n(n+1)}c.
\]
We thus obtain
\[
L(s^*)=nz-c=\frac{2n+1}{n+1}c-c=\frac{n}{n+1}c.\tag*{\qed}
\]

\subsection{Proof of Proposition \ref{p:spence}}\label{s:p3}
First we find the equilibrium wages $w^H$ and $w^L$ after the worker's level of education $e_H$ and $e_L$. For each $j=L,H$, each firm $i$ has the set of speculated states $S_i(e_j)\subset \Omega$. Let this set be the same for each firm. Denote this set by $S(e_j)$, so $S(e_j)=S_1(e_j)=S_2(e_j)$.

 Let $\ul\theta_j$ and $\ol\theta_j$ be the lowest and highest productivity levels given $e_j$, so
\beq\label{e-theta}
\ul\theta_j=\inf\{\theta:(\theta,c)\in S(e_j)\} \quad\text{and}\quad \ol\theta_{j}=\sup\{\theta:(\theta,c)\in S(e_j)\}, \ \ j=L,H.
\eeq
Consider a firm $i$, some wages $w_i$ and $w_{-i}$, and a state $(\theta,c)$. Firm $i$'s maximum profit $u^*_i(w_{-i};\theta)$ is obtained by marginally outbidding $w_{-i}$ when it is below $\theta$, and by choosing the wage below $w_{-i}$ and thus giving up the worker if $\theta\le w_{-i}$, so
\[
u^*_i(w_{-i};\theta)=\sup_{w_i\ge 0 } u_i(w_i,w_{-i};\theta)=\max\{\theta-w_{-i},0\}.
\]
Observe that we only need to consider $w_i$ and $w_{-i}$ in $[\ul\theta_j,\ol\theta_j]$. A wage above $\ol\theta_j$ is dominated and cannot be a best compromise; a wage below $\ul\theta_j$ will always be overbid by the rival's wage, as there is common knowledge that $\theta\ge \ul\theta_j$.

Suppose that $w_i<w_{-i}$, so $u_i(w_i,w_{-i};\theta)=0$.  Then the largest loss is obtained when $\theta$ is the greatest:
\[
\sup_{\theta:(\theta,c)\in S(e_j)} (u_i^*(w_{-i};\theta)-u_i(w_i,w_{-i};\theta))\le \max\{\ol \theta_{j}-w_{-i},0\}.
\]
Next, suppose that $w_i>w_{-i}$, so $u_i(w_i,w_{-i};\theta)=\theta-w_i$. Then the largest loss is obtained when $\theta$ is the smallest:
\[
\sup_{\theta:(\theta,c)\in S(e_j)} (u_i^*(w_{-i};\theta)-u_i(w_i,w_{-i};\theta))=\max\{\theta-w_{-i},0\}-(\theta-w_i)\le w_i-\ul \theta_{j}.
\]
Finally, suppose that $w_i=w_{-i}$, so $u_i(w_i,w_{-i};\theta)=(\theta-w_i)/2$. Then
\begin{multline*}
\sup_{\theta:(\theta,c)\in S(e_j)} (u_i^*(w_{-i};\theta)-u_i(w_i,w_{-i};\theta))=\max\{\theta-w_{-i},0\}-\frac{\theta-w_i}2\\ \le  \max\{0,\ol \theta_{j}-w_{-i}, (w_i-\ul \theta_{j})/2\}.
\end{multline*}
The maximum loss $l_i(w_i,w_{-i})$ is given by the greatest of the three expressions, so
\[
l_i(w_i,w_{-i})=\max\{0,\ol \theta_{j}-w_{-i},w_i-\ul \theta_{j}.\}.
\]
The wages $w_i$ that minimizes the maximum loss satisfies 
\[
w_i=\ol \theta_{j}+\ul \theta_{j}-w_{-i}, \ \ i=1,2.
\]
So, we have obtained two equations, one for each $i=1,2$. Solving this pair of equations for $w_1$ and $w_2$ yields the best compromise $w^*_i(e_j)$ for each firm $i$, where
\beq\label{e-spence-w}
w_i^*(e_j)=\frac{\ol \theta_{j}+\ul \theta_{j}}{2}, \ \ i=1,2.
\eeq
The associated maximum losses are
\beq\label{e-spence-c}
l_i(w^*_i(e_j),w^*_{-i}(e_j))=w^*_i(e_j)-\ul \theta_{j}.
\eeq

Next, observe that the worker operates under complete information. Given each choice of $e_j$, she anticipates the wages $w^j=w^*_1(e_j)=w^*_2(e_j)$, $j\in\{L,H\}$. So, given a state $(\theta,c)$, the worker chooses $e=e_H$ if and only if\footnote{The tie breaking is arbitrary, because the set of types is a continuum.}
\[
w^H-c(\theta)\ge w^L.
\]
Recall that $c(\theta)$ is strictly decreasing, and denote by $c^{-1}$ its inverse. Then, the worker chooses $e=e_H$ if and only if her type $\theta$ satisfies
\[
\theta\ge c^{-1}(w^H-w^L).
\]

{\it Pooling PCE.} If $w^H\le w^L$, then every type chooses low level of education $e_L$, so the equilibrium is pooling.
After observing $e=e_L$, the consistent set of speculated states $S(e_L)$ is thus the entire set of states, so $S(e_L)=\Omega$. By \eqref{e:spence-types}, the highest and lowest $\theta$ in $S(e_L)$ are $\ol\theta_{L}=1$ and $\ul\theta_{L}=0$. By \eqref{e-spence-w}, we obtain the equilibrium wages $w_i(e_L)=1/2$. After observing an out-of-equilibrium education $e=e_H$, the set of speculated states $S(e_H)$ must induce the wage $w^*_i(e_H)\le w^*_i(e_L)$. In particular, we can assume $S(e_H)=\Omega$, and thus $w^*_i(e_H)=1/2$.

Substituting the wage of $w_i^*(e)=1/2$ and the lower bound productivity $\ul \theta_{L}=0$ into  \eqref{e-spence-c}, we obtain the maximum loss for each firm $i$,
\[
l_i(w^*_i(e_j),w^*_{-i};e_j)=\frac 1 2, \ \ i=1,2, \ \ j=L,H.
\]

{\it Separating PCE.} Consider now $w^H>w^L$, so that the worker with cost $c\le w^H-w^L$ chooses high education. Let 
\[
S(e_L)=\{(\theta,c)\in\Omega: c>w^H-w^L\} \quad\text{and}\quad S(e_H)=\{(\theta,c)\in\Omega: c(\theta)\le w^H-w^L\}
\]
be the sets of beliefs of each firm when the level of education is $e_L$ and $e_H$, respectively. So, $S(e_L)$ and $S(e_H)$ contain all pairs $(\theta,c)$ such that low and high education is chosen, respectively. These sets thus satisfy the consistency requirement (Definition \ref{def:consistency}).

By \eqref{e:spence-types} and \eqref{e-theta}, the highest and lowest $\theta$ in $S(e_H)$ are given by
\beq\label{e-spence-1}
\ol\theta_{H}=1 \quad\text{and}\quad  \ul\theta_{H}=\frac{1-w^H+w^L}b.
\eeq
Similarly, the highest and lowest $\theta$ in $S(e_L)$ are given by
\beq\label{e-spence-2}
\ol\theta_{L}=\frac{1+\delta-w^H+w^L}b
 \quad\text{and}\quad  \ul\theta_{L}=0.
\eeq
From \eqref{e-spence-w}, we have
\beq\label{e-spence-3}
w^H=\frac{\ol \theta_{H}+\ul \theta_{H}}{2}\quad\text{and}\quad w^L=\frac{\ol \theta_{L}+\ul \theta_{L}}{2}.
\eeq
Solving the system of six equations in \eqref{e-spence-1}, \eqref{e-spence-2}, and \eqref{e-spence-3}, with six unknowns ($w^H$, $w^L$, $\ol\theta_{H}$, $\ul\theta_{H}$, $\ol\theta_{L}$, and $\ul\theta_{L}$), we obtain the equilibrium wages and the bounds on the productivity types as shown in \eqref{e-spence-w0} and \eqref{e-spence-pb}. 

Observe that the lowest possible cost of high education is $\inf \{c:(\theta,c)\in\Omega\}=1-b$. Therefore, there exist states $(\theta,c)$ where high education $e_H$ is chosen if and only if $w^H-w^L>1-b$. Substituting our solution for $w^H$ and $w^L$ given by \eqref{e-spence-w0}, we obtain that  $w^H-w^L>1-b$ if and only if
\[
\delta<2b^2-b.
\]
This condition is thus necessary and sufficient for the existence of separating PCE.

Finally, substituting the wage $w^H$ and the productivity lower bound  $\ul\theta_{H}$ into \eqref{e-spence-c}, we obtain
firm $i$'s maximum loss when $e=e_H$,
\[
l_i(w^*_i(e_H),w^*_{-i}(e_H);e_H)=w^H-\ul\theta_{H}=\frac 1 2-\frac{b+\delta}{4b^2}.
\]
Substituting the wage $w^L$ and the productivity lower bound  $\ul\theta_{L}$  into \eqref{e-spence-c}, we obtain
the maximum loss when $e=e_L$,
\[
l_i(w^*_i(e_L),w^*_{-i}(e_L);e_L)=w^L-\ul\theta_{L}=\frac{\delta}{2b}+\frac{b+\delta}{4b^2}.\tag*{\qed}
\]

\subsection{Proof of Proposition \ref{p:trade-b}}\label{s:p4}

Consider how a buyer who knows that $v$ is in $[y_0,y_1]$ reacts when the seller asks $p$. Suppose that $p<1/2$. Then the buyer speculates that $v$ in $\{y_0\}$. This is consistent with the strategy of the seller as $p<1/2$ is out of equilibrium. Given this speculation, accepting $p$ if and only if $p\le y_0$ is a best compromise.

Now suppose that $p\ge 1/2$. The largest interval $[x_0,x_1]\subset [0,1]$ that satisfies \eqref{e-p4-1} is $[2p-1,1]$. So the buyer concludes that
\[
v\in V_b(p,y_0,y_1)=[y_0,y_1]\cap [2p-1,1]=[\max\{y_0,2p-1\},y_1].
\]
Given this information about the set of possible values, the buyer now compares her maximum losses when accepting ($\alpha=1$) and rejecting ($\alpha=0$) the price $p$. The maximum loss from rejecting $p$ is
\[
l_b(0;p,y_0,y_1)=\sup_{v\in [\max\{y_0,2p-1\},y_1]}(v-p)=y_1-p.
\]
The maximum loss from accepting $p$ is
\[
l_b(1;p,y_0,y_1)=\sup_{v\in [\max\{y_0,2p-1\},y_1]}(p-v)=\min\{p-y_0,1-p\}.
\]
Because $y_1\le 1$, it is easy to verify that $l_b(0;p,y_0,y_1)\ge l_b(1;p,y_0,y_1)$ if and only if $p\le\frac 1 2(y_0+y_1)$. Thus, it is a best compromise to buy the good when $p\le\frac 1 2(y_0+y_1)$ and not to buy it otherwise.

Let us consider the first stage of the game. Anticipating the buyer's equilibrium behavior $\alpha^*$, the seller chooses a price that minimizes his maximal loss. Observe that choosing a price $p<1/2$ is dominated by $p=1/2$. This is because when $p<1/2$, the buyer accepts $p$ if and only if the value $v$ is guaranteed to be at least as high as the price $p$. In this case, the seller's payoff cannot be positive.

Let $p\ge 1/2$. Suppose first that $p>\frac 1 2(y_0+y_1)>v$. So $p$ is rejected, but it would be optimal to reduce the price so that the buyer accepts it, specifically, to ask $p'=(y_0+y_1)/2$, and thus gain $p'-v$. The supremum of this loss is given by
\[
\sup_{\substack{(v,y_0,y_1):\, p>\frac 1 2(y_0+y_1)>v,\\ v\in[x_0,x_1]\cap[y_0,y_1]}} \left(\frac{y_0+y_1} 2-v\right)=p-x_0.
\]
Second, suppose that $p\le \frac 1 2(y_0+y_1)<v$. So $p$ is accepted, but it would be optimal not to sell, and thus gain $v-p$. The supremum of this loss is given by
\[
\sup_{\substack{(v,y_0,y_1):\, p\le \frac 1 2(y_0+y_1)<v,\\ v\in[x_0,x_1]\cap[y_0,y_1]}} \left(v-p\right)=x_1-p.
\]
Third, suppose that $p\le \frac 1 2(y_0+y_1)$ and $v\le \frac 1 2(y_0+y_1)$. So $p$ is accepted, but it would be optimal to sell at a higher price, specifically, at $p'=\frac 1 2(y_0+y_1)$, and thus gain $p'-p$. The supremum of this loss is given by
\[
\sup_{\substack{(v,y_0,y_1):\, p,v\le \frac 1 2(y_0+y_1),\\ v\in[x_0,x_1]\cap[y_0,y_1]}} \left(\frac{y_0+y_1} 2-p\right)=\frac{x_1+1}{2}-p.
\]
Finally, suppose that $p>\frac 1 2(y_0+y_1)$ and $v\ge \frac 1 2(y_0+y_1)$. So, $p$ is rejected, but any price $p'>v$ would have been rejected too, so the loss is zero in this case.

The maximum loss associated with the price $p\ge 1/2$ is the largest of the four losses computed above, so
\[
l_s(p;x_0,x_1)=\max\left\{p-x_0,x_1-p,\frac{x_1+1}{2}-p,0\right\}=\max\left\{p-x_0,\frac{x_1+1}{2}-p\right\}.
\]
A best compromise price minimizes the maximum loss $l_s(p;x_0,x_1)$ among all prices $p\ge 1/2$, leading to the seller's equilibrium strategy \eqref{e-p4-1}.
\qed

\subsection{Proof of Proposition \ref{p:for1}}\label{s:P-F}
Before proving Proposition \ref{p:for1}, we present a simple lemma on how the loss of a forecast is computed.  
\begin{lemma}\label{l:quadratic}
$
l(a;z)=\sup_{F\in\mathcal  F_{\delta}}(a-\E_{F,G_{\eps}}[\theta|z])^2$.
\end{lemma}

The intuition is as follows. The variance of $\theta$ conditional on a signal $z$ enters the payoffs additively, and thus cancels out when computing the loss. As a result, the maximum loss $l(a;z)$ is simply the maximum quadratic distance between a forecast $a$ and the mean value of $\theta$ conditional on $z$.

\begin{proof}[Proof of Lemma \ref{l:quadratic}]
Fix $G_{\eps}$. Let $\ol a_F(z)=\E_{F,G_{\eps}}[\theta|z]$. Observe that 
\beq\label{e:l-f}
\ol a_F(z)\in\argmax\limits_{a'\in[0,1]} \E_{F,G_{\eps}}[-(a'-\theta)^2|z].
\eeq
So, we have
\begin{multline*}
\sup\limits_{a'\in[0,1]} \E_{F,G_{\eps}}[-(a'-\theta)^2|z]-\E_{F,G_{\eps}}[-(a-\theta)^2|z]= \E_{F,G_{\eps}}[-(\ol a_F(z)-\theta)^2+(a-\theta)^2|z]\\
=\E_{F,G_{\eps}}[(a-\ol a_F(z))(a+\ol a_F(z)-2\theta)|z]=(a-\ol a_F(z))^2,
\end{multline*}
where the first equality is by \eqref{e:l-f} and the last equality is by $\E_{F,G_{\eps}}[\theta|z]=\ol a_F(z)$. Thus,
\[
l(a;z)=\sup_{F\in\mathcal  F_{\delta}} (a-\ol a_F(z))^2=\sup_{F\in\mathcal  F_{\delta}} (a-\E_{F,G_{\eps}}[\theta|z])^2.\qedhere
\]
\end{proof}

We now prove Proposition \ref{p:for1}. Different distributions $F\in\mathcal  F_{\delta}$ induce different conditional means $\E_{F,G_{\eps}}[\theta|z]$. Let $H(z)$ and $L(z)$ be the highest and lowest conditional means, respectively, so 
\beq\label{e:ex5-0}
H(z)=\sup_{F\in\mathcal  F_{\delta}}  \E_{F,G_{\eps}}[\theta|z] \quad\text{and}\quad L(z)=\inf_{F\in\mathcal  F_{\delta}}  \E_{F,G_{\eps}}[\theta|z].
\eeq
The loss of a forecast $a$ given a signal $z$ is 
\[
l(a;z)=\sup_{F\in\mathcal  F_{\delta}} (a-\E_{F,G_{\eps}}[\theta|z])^2=\max\left\{(a-H(z))^2,(a-L(z))^2\right\}
\]
where the first equality is by Lemma \ref{l:quadratic}, and the last equality is by the convexity of the expression. Thus, the best compromise forecast is the midpoint between the highest and lowest conditional means, so
\[
a^*(z)=\inf_{a\in[0,1]}l(a;z)=\frac{1}{2}\left(H(z)+L(z)\right).
\]

It remains to find $H(z)$ and $L(z)$. Suppose that $z\ge\theta_0$.
Observe that
\[
\E_{F,G_{\eps}}[\theta|z]=\frac{(1-\eps)f(z)z+\eps\int_0^1\theta f(\theta)\df\theta}{(1-\eps)f(z)+\eps\int_0^1 f(\theta)\df\theta}=\frac{(1-\eps)f(z)z+\eps\theta_0}{(1-\eps)f(z)+\eps}
\]
is increasing in $f(z)$. Using the assumption that $f(z)\le 1/\delta$, we have
\[
H(z)=\sup_{F\in\mathcal  F_{\delta}}\frac{(1-\eps)f(z)z+\eps\theta_0}{(1-\eps)f(z)+\eps}=\left.\frac{(1-\eps)f(z)z+\eps\theta_0}{(1-\eps)f(z)+\eps}\right|_{f(z)=1/\delta}=\frac{(1-\eps)z+\eps\delta\theta_0}{1-\eps+\eps\delta}.
\]
Using the assumption that $f(z)\ge \delta$, we have
\[
L(z)=\inf_{F\in\mathcal  F_{\delta}}\frac{(1-\eps)f(z)z+\eps\theta_0}{(1-\eps)f(z)+\eps}=\left.\frac{(1-\eps)f(z)z+\eps\theta_0}{(1-\eps)f(z)+\eps}\right|_{f(z)=\delta}=\frac{(1-\eps)\delta z+\eps\theta_0}{(1-\eps)\delta+\eps}.
\]
Analogously,  for $z\le\theta_0$ we obtain $H(z)=\frac{(1-\eps)\delta z+\eps\theta_0}{(1-\eps)\delta+\eps}$ and $L(z)=\frac{(1-\eps)z+\eps\delta\theta_0}{1-\eps+\eps\delta}$. Thus we obtain
\[
a^*(z)=\frac{1}{2}\left(H(z)+L(z)\right)=\frac{1}{2}\left(\frac{(1-\eps)z+\eps\delta\theta_0}{1-\eps+\eps\delta}+\frac{(1-\eps)\delta z+\eps\theta_0}{(1-\eps)\delta+\eps}\right).\qedhere
\]

\section*{Appendix B. Alternative Model of Forecasting}\label{s:F-2}
\renewcommand{\thesection}{B}
This section considers an alternative variation of the forecasting model presented in Section \ref{s:forec}. Here we are interested in how to forecast a random variable with a known distribution after receiving a noisy signal that has an unknown distribution.

Suppose that the agent knows the distribution $F$ of $\theta$, but is uncertain about how the noisy signal $z$ is generated. 
The following assumptions are made about this signal. The signal $z$ is known to be not too far from the true value of $\theta$, where a parameter $\delta>0$ describes the maximal distance. So $\delta$ can also be interpreted as the precision of the signal. Let $y=z-\theta$ be called the noise. So it is known that $|y|\le \delta$. The distribution of the noise $y$ has a certain and an uncertain component.  Let $\eps\in[0,1]$ be a known parameter. With probability $1-\eps$ the noise $y$ is drawn from a known distribution $G_0$ and with probability $\eps$ it is drawn from an unknown distribution $G_1$. So $\eps$ measures how uncertain the agent is about how the noise is generated. Given the support restrictions on $y$, it follows that $G_0$ and $G_1$ both have support contained in $[-\delta, \delta]$. Let $G_\delta$ be the set of all distributions of $y$ that satisfy the above description.

Let $\E_{F,G_\delta,\eps}[\cdot|z]$ denote the conditional mean of $\theta$ given $z$ for $G_\delta\in \mathcal G_\delta$. The maximum loss associated with a forecast $a\in[0,1]$ given a signal $z\in[0,1]$ is 
\[
l(a;z)=\sup_{G_\delta\in\mathcal  G_{\delta}} \left(\sup_{a'\in[0,1]} \E_{F,G_\delta,\eps}[-(a'-\theta)^2|z]-\E_{F,G_\delta,\eps}[-(a-\theta)^2|z]\right).
\]
Let $H(z)$ and $L(z)$ be the highest and lowest conditional means, so 
\[
H(z)=\sup_{G_\delta\in\mathcal  G_{\delta}} \E_{F,G_\delta,\eps}[\theta|z] \quad\text{and}\quad L(z)=\inf_{G_\delta\in\mathcal  G_{\delta}}\E_{F,G_\delta,\eps}[\theta|z].
\]
It is straightforward to verify that
\[
H(z)=\sup_{x\in[-\delta,\delta]} \frac{\eps f(z-x)(z-x)+(1-\eps)\int_{-\delta}^{\delta}(z-y) f(z-y)\df G_0(y)}{\eps f(z-x)+(1-\eps)\int_{-\delta}^{\delta}f(z-y)\df G_0(y)},
\]
with an analogous expression for $L(z)$. We obtain the following result.

\begin{proposition}
The agent's best compromise is
\[
a^*(z)=\frac 1 2\left(H(z)+L(z)\right).
\]
\end{proposition}

The proof is analogous to that of Proposition \ref{p:for1} and thus omitted.

The best compromise is the midpoint between the highest and lowest conditional means. The agent's best compromise forecast depends on the precision $\delta$ of her signal, as well as on the degree $\eps$ of her uncertainty. We show how each of these two parameters independently influences the best compromise forecast.

Fix the degree of uncertainty $\eps$. If the signal is very precise in the sense that $\delta$ is very small, then each of the two extreme conditional means are close to $z$. Hence, the best compromise forecast will also be close to $z$. Formally, $\lim_{\delta\to 0} a^*(z)=z$. 

Fix the precision $\delta$ of the signal. As the degree of uncertainty $\eps$ vanishes, both extreme conditional means converge to the conditional mean under the benchmark distribution $G_0$. Formally, $\lim_{\eps\to 0} a^*(z)= \E_{F,G_0,0}[\theta|z]$. For instance, if $G_0$ is the uniform distribution, then the best compromise forecast converges to the expected value of $\theta$ conditional on $\theta$ being within $\delta$ of the signal.

As the degree of uncertainty $\eps$ becomes large, the role of the benchmark $G_0$ diminishes and almost any noise within $[-\delta,\delta]$ becomes possible. When $\eps=1$, it could be that $G_1$ puts all mass on $-\delta$, in which case $\E_{F,G_\delta,\eps}[\theta|z]=z+\delta$. This is the highest conditional mean given $z$, so $H(z)=z+\delta$. It could also be that  $G_1$ puts all mass on $\delta$, in which case $\E_{F,G_\delta,\eps}[\theta|z]=z-\delta$. This is the lowest conditional mean given $z$, so $L(z)=z-\delta$. Consequently, the best compromise forecast is close to the signal $z$ when the agent is very uncertain about how $z$ is generated. Formally, $a^*(z)\to z$ as $\eps\to 1$. 

\begin{remark}\label{R:F}
Note that the distribution $F$ of the underlying variable of interest plays no role when the degree of uncertainty is extreme, so $\eps=1$. Consequently, we obtain that if the agent knows neither $F$ nor the distribution of the noise, then the best compromise forecast is to choose the signal.
\end{remark}

{\setlength{\baselineskip}{0.2in} 
\setlength\bibsep{0.2\baselineskip}
\bibliographystyle{aea}
\bibliography{bc}}

\end{document}